\documentclass[runningheads,smallsec]{llncs}
\newcommand*{\N}{\mathbb{N}}

\newcommand*{\M}{\mathcal{M}}
\newcommand*{\B}{\mathcal{B}}

\newcommand{\expspace}{\textsf{EXPSPACE}}
\newcommand{\pspace}{\textsf{PSPACE}}

\newcommand{\ltlf}{\mathsf{LTLf}}
\newcommand{\ltl}{\mathsf{LTL}}

\newcommand*{\size}[1]{|#1|}

\newcommand*{\A}{\mathcal{A}}

\newcommand*{\Final}{\mathcal{F}}

\newcommand*{\bigO}{\mathcal{O}}

\newtheorem{claim}{Claim}

\newcommand{\ltlfU}{\text{U}}
\newcommand{\ltlfW}{\text{W}}
\newcommand{\ltlfX}{\text{X}}
\newcommand{\ltlfR}{\text{R}}
\newcommand{\ltlfN}{\text{N}}
\newcommand{\ltlfNeg}{\neg}
\newcommand{\ltlfG}{\text{G}}
\newcommand{\ltlfF}{\text{F}}
\newcommand{\ltlG}{\text{G}_{\infty}}
\newcommand{\ltlF}{\text{F}_{\infty}}
\newcommand{\ltlU}{\text{U}_{\infty}}

\newcommand{\ltlftrue}{\textsf{true}}
\newcommand{\ltlffalse}{\textsf{false}}
\newcommand{\nnf}{\textsf{nnf}}

\newcommand{\alphabet}{\Sigma}
\newcommand{\states}{S}

\newcommand{\init}{\iota}

\newcommand{\lang}[1]{\mathcal{L}(#1)}

\newcommand{\ap}{\mathsf{Prop}}

\newcommand{\pref}{\mathsf{pref}}

\newcommand{\setnocond}[1]{\{#1\}}

\iffalse
	\usepackage{todonotes}
	
	\newcommand{\ly}[1]{\todo[inline,color=orange!10,caption={LY}]{\textbf{LY:} #1}}
	\newcommand{\lmt}[1]{\todo[inline,color=red!10,caption={LMT}]{\textbf{LMT:} #1}}
	\newcommand{\myv}[1]{\todo[inline,color=green!10,caption={MYV}]{\textbf{MYV:} #1}}
	\newcommand{\aw}[1]{\todo[inline,color=yellow!10,caption={AW}]{\textbf{AW:} #1}}
\else
	\renewcommand{\sb}[1]{}
	\newcommand{\ly}[1]{}
	\newcommand{\lmt}[1]{}
	\newcommand{\myv}[1]{}
	\newcommand{\aw}[1]{}
	
\fi

\usepackage{algorithm}
\usepackage[noend]{algorithmic}
\usepackage{amsmath}
\usepackage{amsfonts}
\usepackage{amssymb}
\usepackage{comment}
\usepackage{enumitem}
\usepackage{graphicx}
\usepackage[bookmarks=false]{hyperref}
\usepackage{listings}
\usepackage{cite}
\usepackage{makeidx}  
\usepackage{multicol}
\usepackage{subcaption}
\usepackage{url}
\usepackage{wrapfig}

\usepackage[graphicx]{realboxes}
\usepackage{tikz}
\tikzset{elliptic state/.style={draw,ellipse}}
\usetikzlibrary{shapes,shapes.geometric,arrows,fit,calc,positioning,automata,}
\usetikzlibrary{automata,positioning}
\usepackage{tikz-qtree}
\usepackage{float}
\usepackage{floatflt}

\begin{document}

\title{Model Checking Strategies from Synthesis Over Finite Traces}
%
%

 \author{Suguman Bansal\inst{1}\orcidID{0000-0002-0405-073X} \and  Yong Li\inst{2}\orcidID{0000-0002-7301-9234} \and Lucas M. Tabajara\inst{3}\orcidID{0000-0001-9608-1404}, Moshe Y. Vardi\inst{4}\orcidID{0000-0002-0661-5773} \and Andrew Wells\inst{4}\orcidID{0000-0001-7780-2122}\thanks{Work was performed while the author was at Rice University}}
 \authorrunning{Bansal et. al.}
 \institute{Georgia Institute of Technology, GA, USA
 \and
 University of Liverpool, UK
 \and
 Runtime Verification, USA
 \and
 Rice University, TX, USA
 }

\maketitle              
\begin{abstract}

The innovations in reactive synthesis from {\em Linear Temporal Logics over finite traces} ($\ltlf$) will be amplified by the ability to  verify the correctness of the strategies generated by $\ltlf$ synthesis tools. This motivates our work on {\em $\ltlf$ model checking}.
$\ltlf$ model checking, however, is not  straightforward. The strategies generated by $\ltlf$ synthesis may be represented using  {\em terminating} transducers or {\em non-terminating} transducers where executions are of finite-but-unbounded length or infinite length, respectively.
For synthesis, there is no evidence that one type of transducer is better than the other since they both demonstrate the same complexity and similar algorithms. 


In this work, we show that for model checking, the two types of transducers are fundamentally different. Our central result is that $\ltlf$ model checking of non-terminating transducers is \emph{exponentially harder} than that of terminating transducers.
We show that the problems are \linebreak\expspace-complete and $\pspace$-complete, respectively.  Hence, considering  the feasibility of verification,  $\ltlf$ synthesis tools should synthesize terminating transducers. This is, to the best of our knowledge, the \emph{first} evidence to use one transducer over the other in $\ltlf$ synthesis.

\end{abstract}

\sloppy
\section{Introduction}
\label{Sec:Intro}

{\em Linear Temporal Logic over finite traces}~\cite{de2013linear} ($\ltlf$) is the finite-horizon counterpart of the well-known Linear Temporal Logic ($\ltl$) over infinite traces~\cite{pnueli1977temporal}.
$\ltlf$ is rapidly gaining popularity among real-world applications where behaviors are better expressed over a finite but unbounded horizon~\cite{brafman2019planning,de2019foundations,de2018automata,he2017reactive,wells2020ltlf}. 

Reactive synthesis from $\ltlf$ specifications, or $\ltlf$ synthesis ~\cite{bansal2020hybrid,camacho2019ltl,de2022ltlf,de2015synthesis,de2016ltlf,favorito2023forward,tabajara2019partition,zhu2017ltlfsymbolic} has amassed so much interest that the 2023 Reactive Synthesis Competition (SYNTCOMP) will inaugrate an $\ltlf$ track\footnote{http://www.syntcomp.org/news/}. Consequently,
$\ltlf$ synthesis tools have been  growing in complexity~\cite{bansal2020hybrid,de2021compositional,favorito2023forward,tabajara2019partition,zhu2017ltlfsymbolic}. Their correctness, however, is rarely verified. To continue the innovations in synthesis  and to successfully conduct large-scale competitions like SYNTCOMP there is, therefore, a need to verify the correctness of the synthesized strategies/transducers. Verifying the results as opposed to verifying the tools has been advocated in various contexts, including translation validation~\cite{siegel1998translation}, program checking~\cite{blum1995designing}, and equivalence checking~\cite{kuehlmann2002combinational}. 
For $\ltl$ synthesis, result checking is simply {\em $\ltl$ model checking}.
For $\ltlf$ synthesis, we need {\em $\ltlf$ model checking}. But this is a topic that has \emph{not} been studied so far, hence this work.

We observe that $\ltlf$ model checking for $\ltlf$ synthesis tools is \emph{not} as straightforward as one might have thought to be. The standard approach in the literature on $\ltlf$ synthesis generates {\em non-terminating transducers}.
This includes the seminal work on synthesis \cite{de2015synthesis} and the SYNTCOMP guidelines~\cite{jacobs2023temporal}.
The executions of non-terminating transducers are of infinite length. 
Since $\ltlf$ formulas are defined on finite traces only, an execution of a non-terminating transducer is said to satisfy an $\ltlf$ formula if there {\em exists} a finite-length prefix that satisfies the formula~\cite{de2015synthesis}.
Few works on  synthesis do mention the possibility of {\em terminating transducers} as the output \cite{bansal2020hybrid,zhu2017ltlfsymbolic}. Since their executions are of finite length, $\ltlf$ satisfaction is defined naturally on terminating transducers.
When it comes to synthesis, there is no clear evidence that one type of transducer is better than the other, since the complexity and algorithms of synthesis are the same for both types.
We believe this is why existing works on $\ltlf$ synthesis do \emph{not} make a clear distinction between the two. For implementations, however, most works use non-terminating transducers as they directly correspond to standard Mealy/Moore machines (See state-of-the-art tools, e.g., Syft~\cite{zhu2017ltlfsymbolic}, Lisa~\cite{bansal2020hybrid}, and Lydia~\cite{de2021compositional}).
This work shows, however, that from the \emph{model-checking} perspective, the two types of transducers are \emph{fundamentally different} and bear a significant impact on  synthesis.

Our central result is that $\ltlf$ model checking of non-terminating transducers is {\em exponentially harder} than $\ltlf$ model checking of terminating transducers. 
We demonstrate that under $\ltlf$ specifications,  model checking non-terminating transducers is $\expspace$-complete, whereas model checking terminating transducers is $\pspace$-complete.
An immediate implication of this result is that for non-terminating transducers, $\ltlf$  model checking is exponentially harder than $\ltl$ model checking, which is  known to be $\pspace$-complete~\cite{vardi1986automata}.  
This result is unexpected because a factor behind the increasing popularity of $\ltlf$ is the perception that problems using $\ltlf$ are at most as hard as those using $\ltl$, if not simpler (See~\autoref{tab:complexity}). This is because $\ltlf$ formulas can be expressed by automata over finite words~\cite{de2013linear}, which allow for practically scalable algorithms for automata constructions~\cite{TRV12}. Conversely, $\ltl$ formulas require automata over infinite words~\cite{wolper1983reasoning}, for which the automata manipulation is harder in  theory~\cite{esparza2020unified,safra1988complexity,thomas2002automata,vardi2007buchi} and in practice~\cite{duret2016spot,kvretinsky2018owl}.
It is no wonder that an exponential increase in the model-checking complexity seems surprising at first.

The exponential blow-up in $\ltlf$ model-checking of non-terminating transducers arises from subtlety in the problem definition. 
A transducer satisfies a formula if there are no counterexamples. 
In non-terminating transducers, an infinite execution is a counterexample if {\em every} finite prefix does not satisfy the $\ltlf$ formula.
Formally, for an $\ltlf$ formula $\phi$, let $\pref(\phi)$ represent the language consisting of all  infinite executions for which every prefix satisfies $\phi$. Then, a non-terminating transducer $\M$ satisfies an $\ltlf$ formula $\phi$ iff $\lang{\M} \cap \pref(\neg\phi) = \emptyset$, where $\lang{\M}$ is the set of all executions of $\M$.
This is where $\ltlf$ model checking fundamentally differs from $\ltl$ model checking, as counterexamples in $\ltl$ are obtained simply from an automaton for the negation of the  formula~\cite{vardi1986automata}. 
W.l.o.g., we show that while $\pref(\phi)$ is $\omega$-regular for all $\ltlf$ formulas $\phi$, the size of their non-deterministic B\"uchi automata (NBA) is {doubly exponential} in the size of the formula, i.e., $2^{2^{\bigO(|\phi|)}}$ and $2^{2^{\Omega(\sqrt{|\phi|})}}$. Once again, this differs from $\ltl$ model checking, where the size of the NBAs for counterexamples is singly exponential in the size of the formula.  
As a result, we show $\ltlf$ model checking  of non-terminating transducers is in $\expspace$ using on-the-fly emptiness checking of  $\lang{\M}\cap \pref(\neg\phi)$.
We establish $\expspace$-hardness from  first principles.

In contrast, we show that $\ltlf$ model checking of terminating transducers is \pspace-complete.  Due to their finite-length executions, counterexamples in terminating transducers are  completely characterized by the negation of the  formula, lending the same complexity as $\ltl$ model checking.

Thus, our results offer a clear recommendation between  the two types of transducers in $\ltlf $ synthesis. 
We argue that synthesis tools should account for the feasibility of the verification of the synthesized transducers. Consequently, we recommend that synthesis tools should generate terminating transducers rather than non-terminating transducers. We believe this is the \emph{first} work to offer \emph{theoretical} evidence to use one transducer over the other in synthesis. Furthermore, these results could be applied immediately to run the $\ltlf$ track in SYNTCOMP. 


\begin{table}[t]
\centering
\caption{$\ltl$ vs. $\ltlf$: Complexity w.r.t. specification.  NT and T abbreviate non-terminating and terminating models, respectively.
}
\begin{tabular}{c | c | c }
\hline
 & $\ltl$    & $\ltlf$   \\ \hline

Non-deterministic Automata & (NBA) Exponential & (NFA) Exponential \\ 
Satisfiability      & ~${\pspace}$-complete~\cite{sistla1985complexity} ~ & ${\pspace}$-complete~\cite{de2013linear}  \\ 
Synthesis & {$\textsf{2EXPTIME}$}-complete~\cite{pnueli1989synthesis} & ~~{$\textsf{2EXPTIME}$}-complete~\cite{de2015synthesis}  \\ 
Model Checking (NT) ~& ${\pspace}$-complete ~\cite{vardi1986automata} & {\color{blue} ${\expspace}$-complete (New!)}  \\ 
Model Checking (T) ~& Undefined & {\color{blue} ${\pspace}$-complete (New!)} \\ 
\hline
\end{tabular}
\label{tab:complexity}
\end{table}

\subsubsection*{Outline.}
Section~\ref{sec:prelims} outlines preliminaries on $\ltlf$ and $\ltlf$ synthesis. Section~\ref{Sec:probstatement} motivates and defines $\ltlf$
model checking. Section~\ref{sec:prefix-automata} is dedicated to $\pref(\phi)$. Section~\ref{sec:complexity} develops the  complexity of model checking. Lastly, Section~\ref{sec:discussion} concludes. 

\section{Preliminaries and Notations}
\label{sec:prelims}


We use the standard notions of deterministic and non-deterministic finite automata (DFAs and NFAs, respectively) as well as deterministic and non-deterministic B\"uchi automata (DBAs and NBAs, respectively). For an automaton, we use the notation $\A = (\alphabet, S , \init, \delta, F )$ 
where $\alphabet$ is a finite set of symbols (called an alphabet),
$S$ is a finite set of states,
$\init \in \states$ is the initial state,
$F \subseteq S$ is the set of accepting states, and
$\delta \subseteq S \times \Sigma \times S $ is the transition relation. We use standard semantics for all automata, hence defer details to the appendix. 

\subsection{Linear Temporal Logic over Finite Traces ($\ltlf$)}
\label{sec:ltlf-def}

$\ltlf$~\cite{baier2006planning,de2013linear} extends propositional logic with finite-horizon temporal operators. In effect, $\ltlf$  is a variant of $\ltl$~\cite{pnueli1977temporal} that is interpreted over finite rather than infinite traces.
The syntax of an $\ltlf$ formula over a finite set of propositions $\ap$ is identical to $\ltl$, and defined as 
$$\varphi := \mathsf{true} \mid \mathsf{false} \mid a \in \ap \mid \neg \varphi \mid \varphi_1 \land \varphi_2 \mid \ltlfX \varphi \mid \varphi_1 \ltlfU \varphi_2 $$
where $\ltlfX$ (Next) and $\ltlfU$ (Until),   are temporal operators. We also include their dual operators, $\ltlfN$ (Weak Next) and $\ltlfR$ (Release), defined as $\ltlfN \varphi 
 \equiv \neg\ltlfX \neg\varphi$ and $\varphi_1\ltlfR\varphi_2 \equiv \neg(\neg\varphi_1 \ltlfU \neg\varphi_2)$.
 We also use typical abbreviations such as $\ltlfF \varphi \equiv \mathsf{true}\ltlfU \varphi$,  $\ltlfG \varphi \equiv  \mathsf{false} \ltlfR \varphi$, $\varphi_1\vee \varphi_2 = \neg(\neg \varphi_1 \land \neg\varphi_2)$, $\varphi_1 \rightarrow \varphi_2 \equiv \neg \varphi_1 \lor \varphi_2$.
 We denote by $\size{\phi}$ the length/size of a formula $\phi$, i.e., the number of operators in $\phi$.


The semantics of $\ltlf$ is similar to $\ltl$ but is interpreted over finite traces.
A finite sequence $\rho$ over $2^{\ap}$ is said to satisfy an $\ltlf$ formula $\phi$ over $\ap$, denoted by $\rho\models \phi$, if $\rho, 0 \models \phi$ where for all positions $0 \leq i < \size{\rho}$, $\rho, i\models \phi$ is defined inductively  on $\phi$ as follows:

\begin{itemize}
    \item $\rho, i \models \ltlftrue$; 
    $\rho, i \not\models \ltlffalse$;
    $\rho, i \models a$ iff $a \in \rho_i$
    
    
    \item $\rho, i \models \neg\varphi$ iff $\rho, i \not\models \varphi$
    
    \item $\rho, i \models \phi_{1} \land \phi_{2}$ iff $\rho, i \models \phi_{1}$ and $\rho, i \models \phi_{2}$;
    
    
    \item $\rho, i \models \ltlfX \phi$ iff $i + 1 < \size{\rho}$ and $\rho, i + 1 \models \phi$
    
    \item $\rho, i \models \phi_{1} \ltlfU \phi_{2}$ iff there exists $j$ s.t. $i \leq j < \size{\rho}$ and $\rho, j \models \phi_{2}$, and for all $k$, $i \leq k < j$, we have $\rho, k \models \phi_{1}$
    
\end{itemize}

Observe that $\ltlfX$ requires that there \emph{exists} a next position;
In the context of \emph{finite} traces, its negation also contains the situation that no next position exists, formulated as $\neg (\ltlfX \ltlftrue)$ or equivalently $\ltlfN \ltlffalse$. This differs from $\ltl$ where the Next operator is applied to all positions. 
Also, note that $\ltlf$ formulas are evaluated on traces of non-zero length.  
The language of an $\ltlf$ formula $ \phi $ over $\ap$ is the set of all finite sequences $\rho$  over $2^{\ap}$ such that $\rho \models \phi$.
The language of an $\ltlf$ formula is regular.
The NFA and DFA representing $\ltlf$ are of size singly exponential and doubly exponential, respectively, in the size of the formula~\cite{de2013linear}.
We note that a letter $\sigma \in \Sigma$ of the NFA/DFA corresponds to a valuation over the set $\ap$ of propositions.


\subsection{ $\ltlf$ Synthesis and Transducers}
\label{sec:ltlfsynthesis}

Let  $\ltlf$ formula $\phi$ be defined over propositional variables partitioned into $\mathcal{I}$ and $\mathcal{O}$ representing the input and output variables, respectively. 
Given such an $\ltlf$ formula  $\phi$, the problem of {\em $\ltlf$ realizability} is to determine whether there exists a strategy $f : (2^{\mathcal{I}})^* \rightarrow 2^{\mathcal{O}}$ such that for all $\lambda_\mathcal{I} = I_0, I_1, \cdots \in (2^{\mathcal{I}})^{\omega}$, there is an integer $k \geq 0$ such that the finite trace $\rho =(I_0 \cup f(\varepsilon)),(I_1 \cup f(I_0)), \cdots ,(I_k \cup f(I_0, I_1, \cdots , I_{k-1}))$ satisfies $\phi$. 
The {\em $\ltlf$ synthesis problem} is to generate such a function, if the given formula is realizable~\cite{de2015synthesis}. 
Intuitively, $\ltlf$ synthesis can be viewed as a game between two agents, an environment and a system, who continually take turns to assign values to the input and  output variables, respectively, to generate a sequence of input and output variables. W.l.o.g., we assume the system plays first, followed by the environment, and so on. The goal of synthesis is to generate a strategy for the system agent so that all resulting plays with the environment satisfy the given specification.
We note that our model-checking results also hold when the environment plays first, as we will model strategies as transition systems in model checking for generality (cf. Section~\ref{Sec:probstatement}).

\subsubsection{Non-terminating transducers.}
The standard in $\ltlf$ synthesis is to represent the strategy $f$ using (non-terminating) transducers~\cite{de2015synthesis,jacobs2023temporal}. 
W.l.o.g., a transducer is a {\em Moore machine} $\M = { (Q,q_{0},\mathcal{I} ,\mathcal{O},\delta 
,G)}$ where $Q$ is a finite set of states, $q_0\in Q$ is the initial state, and $\mathcal{I}$ and $\mathcal{O}$ are finite sets of input and output variables, respectively. 
Functions $\delta: Q\times 2^\mathcal{I} \rightarrow Q$ and  $G : Q\rightarrow 2^\mathcal{O}$ are the {\em transition function} and the {\em output function}, respectively. 
Given an input sequence $\lambda_\mathcal{I} = I_0, I_1, \cdots \in (2^{\mathcal{I}})^{\omega}$, the output sequence is  $\lambda_\mathcal{O} = G(q_0),G(q_1),\dots \in (2^{\mathcal{O}})^{\omega}$  where $q_0$ is the initial state and $q_{i+1} = \delta(q_i, I_i)$ for all $i\geq 0$. 

Then, given an $\ltlf$ formula with variables partitioned into $\mathcal{I}$ and $\mathcal{O}$ the realizability and synthesis problem is to generate a Moore machine $\M$ such that for all  input sequences $\lambda = I_0, I_1, \cdots \in (2^{\mathcal{I}})^{\omega}$, there exists an integer $k \geq 0$ such that $\rho = (I_0, G(q_0)), (I_1, G(q_1)) \dots (I_k, G( q_k))$ satisfies $\phi$.
Intuitively, the system and environment play indefinitely, where the system plays as per the transducer. The play (an execution in the transducer) satisfies an $\ltlf$ formula if there exists a finite-length prefix that satisfies the formula. 

\subsubsection{Terminating transducers.}
The strategy $f$ can also be represented using terminating transducers~\cite{bansal2020hybrid,zhu2017ltlfsymbolic}.
W.l.o.g., a terminating transducer is a {\em Terminating Moore machine} $\M = { (Q,q_{0},\mathcal{I} ,\mathcal{O},\delta 
,G, F)}$ where $Q$, $q_0$,  $\mathcal{I}$, $\mathcal{O}$, $\delta$, and $G$ are as defined for Moore machines and $\emptyset \neq F \subseteq Q$ are the {\em terminal states}. 
An input sequence $\lambda_\mathcal{I} = I_0, I_1, \cdots I_k \in (2^{\mathcal{I}})^{*}$ generates an output sequence   $\lambda_\mathcal{O} = G(q_0),G(q_1),\dots G(q_k) \in (2^{\mathcal{O}})^{*}$  where $q_0$ is the initial state and $q_{i+1} = \delta(q_i, I_i)$ for all $0 \leq i < k$.

Then, given an $\ltlf$ formula with variables partitioned into $\mathcal{I}$ and $\mathcal{O}$, the realizability and synthesis problem is to generate a
terminating Moore machine $\M$ such that for all  input sequence $\lambda = I_0, I_1, \cdots \in (2^{\mathcal{I}})^{\omega}$, there exists an integer $k \geq 0$ such that $\rho = (I_0, G(q_0)), (I_1, G(q_1)) \dots (I_k, G( q_k))$ with $q_{k+1} = \delta(q_k, I_k) \in F$ and $\rho$ satisfies $\phi$.
Intuitively, the synthesized terminating transducer is such that as soon as a play lands in a terminal state of the transducer, the system agent controlling the output variables wins the game and this play is over as it is guaranteed that the play seen so far satisfies the given formula. On the contrary, in non-terminating transducers, the system agent does not have the ability to terminate a game
as it is never informed of whether it has seen a satisfying prefix.


\section{$\ltlf$ Model Checking}
\label{Sec:probstatement}

In addition to being of independent interest, our motivation behind  $\ltlf$ model checking is to support the ongoing development of $\ltlf$ synthesis tools. As synthesis tools continue to become more complex, it is imperative that we design automatic approaches to  check their correctness. One way is to evaluate whether the result generated from these tools is correct. In the case of $\ltlf$ synthesis, result checking corresponds to $\ltlf$ model checking. Finally, an immediate application of $\ltlf$ model checking could be in running the inaugural $\ltlf$ track in the Reactive Synthesis Competition (SYNTCOMP)~\cite{jacobs2023temporal}.

We begin by defining the model-checking problem. As described in Section~\ref{sec:ltlfsynthesis}, the result of $\ltlf$ synthesis could be a terminating or a non-terminating transducer. Since $\ltlf$ satisfaction on executions in the two types of transducers differ, we define model-checking on them separately. For the sake of generality, we define model-checking with respect to {\em transition systems} (TS) as opposed to transducers. 
Translations from transducers to transition systems are standard and polynomial~\cite{DBLP:conf/litp/NicolaV90}.
Hence, the translation details have been omitted.

\paragraph{Non-Terminating Transition Systems}
\label{Sec:Def:Nonterminating}
are those that run indefinitely, i.e., their executions are of infinite length (e.g. network servers).
Formally, a non-terminating TS is a structure  $\M = (\Sigma, S, T, \init, L)$, where $\Sigma$ is a finite propositional alphabet, $S$ is a finite set of states, relation $T \subseteq S \times S$ is the transition relation with no sink states, $\init$ is the initial state, and $L: S\rightarrow 2^{\Sigma}$ is the {\em labeling function}. An {\em execution} $\rho = s_0s_1\cdots$  in $\M$ is an infinite sequence of consecutive states beginning with the initial state, i.e., $s_0 =\init$ and $(s_i,s_{i+1}) \in T$ for all $i\geq 0$. 
The {\em label sequence} of $\rho$ is the sequence $L(\rho) = L(s_0)L(s_1)\cdots$. 
The $n$-length finite prefix of $\rho$ and its label sequence are given by $\rho[0,n] = s_0\cdots s_{n-1}$ and $L(\rho[0,n]) = L(s_0)\cdots L(s_{n-1})$, respectively, for $n>0$. 

Since executions are of infinite-length and $\ltlf$ formulas are interpreted over finite-length sequences only, we say an {\em execution $\rho$ in $\M$ satisfies an $\ltlf$ formula} $\phi$, denoted by $\rho \models \M$, as follows $$\rho \models \phi \text{ iff } \exists n> 0 \text{ s.t. } L(\rho[0,n]) \models \phi,$$
i.e., there exists a finite-length prefix of the execution that satisfies the formula. 

\paragraph{Terminating Transition Systems}
\label{Sec:Def:Terminating}
are those that terminate after a finite but unbounded amount of steps (e.g. a terminating program). Formally, a  terminating TS is given by a structure $\M = (\Sigma, S, T, \init, L, F)$, where $\Sigma$, $S$,  $T \subseteq S \times S$, $\init$, and $L: S\rightarrow 2^{\Sigma}$ are defined as for nonterminating transition systems and $\emptyset \neq F \subseteq S$ are the {\em terminal states}, which are the only states that are allowed to be sink states. 
An {\em execution} $\rho = s_0\cdots s_n$  in $\M$ is a finite sequence of consecutive states beginning with the initial state and ending in a terminal state, i.e., $s_0 =\init$ and $(s_i,s_{i+1}) \in T$ for all $0\leq i < n $, and $s_n \in F$. Its {\em label sequence} is the sequence $L(\rho)=L(s_0)\cdots L(s_n)$. 

An {\em execution $\rho$ in $\M$ satisfies an $\ltlf$ formula} $\phi$, denoted by $\rho \models \phi$, 
$$\rho\models \phi \text{ iff } L(\rho) \models \phi.$$

\paragraph{Model Checking.}
We first define \emph{satisfaction} and then \emph{model checking}.

\begin{definition}[$\M\models\phi$]
\label{def:nonterminating}
Given a non-terminating (resp., terminating) transition system $\M$ and an $\ltlf$ formula $\phi$, we say TS {\em $\M$ satisfies $\phi$}, denoted by $\M \models \phi$, if for all (resp., finite) executions $\rho$ of $\M$, we have that $\rho \models \phi$.
\end{definition}

\begin{definition}[Model Checking]
\label{def:terminating}
Given a non-terminating (resp. terminating) transition system $\M$ and an  $\ltlf$ formula  $\varphi$,  the problem of $\ltlf$ model checking of non-terminating (resp. terminating) models is to determine whether  $\M$ satisfies  $\varphi$. 
\end{definition}

\paragraph{Note on abuse of notation.}
The notation $\models$ has been overloaded to express satisfaction at several occasions, namely, in $\ltlf$ semantics, in defining when executions of non-terminating and terminating systems satisfy a formula, and when a system satisfies a formula. We overload notation to avoid new symbols for each case, as the context is clear from the L.H.S.

\section{Prefix Language of $\ltlf$ Formulas}
\label{sec:prefix-automata}

This section builds the basic blocks for $\ltlf$ model checking of non-terminating systems.
Recall from Section~\ref{Sec:probstatement}, an (infinite-length) execution in a non-terminating  system  $\M$ violates an $\ltlf$ formula $\phi$ if {\em all}  of its finite prefixes violate $\phi$.
So,  the counterexamples are captured by the language that accepts an infinite word iff all of its finite prefixes violate $\phi$ (or satisfy $\neg \phi$).
We call this the \emph{prefix language} of an $\ltlf$ formula $\neg \phi$.
Then, clearly,  $\M \models \phi$ iff the intersection of $\M$ with the prefix language of $\neg\phi$ is empty, making the prefix language a basic block to model-check non-terminating systems.

We first observe that the prefix languages for $\ltlf$ formulas are $\omega$-regular.
We then show that one can construct a DBA accepting the prefix language of an $\ltlf$ formula, which incurs a doubly exponential blow-up (Section~\ref{ssec:prefix-construction}).
One may expect that the complexity of the construction can be improved if we target at NBAs.
We show, however, that the doubly exponential blow-up is \emph{not} due to a lack of better construction, but a fundamental trait of the problem itself (Theorem~\ref{thrm:prefixdouble}).
This is in contrast to the construction of NBA/NFA for $\ltl$/ $\ltlf$, where only deterministic automata constructions incur doubly exponential blow-ups and nondeterministic automata constructions incur singly exponential blow-ups, hinting at the hardness of model checking. Finally, we identify a fragment of $\ltlf$ formulas for which a singly exponential construction of NBAs for their prefix languages can be obtained via a translation from $\ltlf$ to $\ltl$ (Section~\ref{ssec:prefix-fragment}).

\subsection{Prefix Automata for $\ltlf$}
\label{ssec:prefix-construction}

This section formally defines the prefix language/automata for $\ltlf$ formulas and proves that their automata constructions involve an unavoidable double-exponential blow-up. 
The upper and lower bounds are shown in Theorem~\ref{thrm:PropertiesPref} and Theorem~\ref{thrm:prefixdouble}, respectively.  

\begin{definition}[Prefix Language]
\label{def:prefixlangauge}
Given an $\ltlf$ formula $\phi$,  the {\em prefix language of $\phi$}, denoted by $\pref(\phi)$, is such that an (infinite-length) word $w \in \pref(\phi)$ iff every finite prefix of $w$ satisfies $\phi$, i.e., 
$\forall n>0. w[0,n] \models \phi$.
\end{definition}
Recall that the semantics of $\ltlf$  requires traces of non-zero length only (see Section~\ref{sec:prelims}).
So we only need $n > 0$, instead of $n \geq 0$, ignoring the empty word.
By abuse of notation, we let $\pref(\phi)$ denote both the prefix language and its corresponding automaton, called the \emph{prefix automaton}.
 
 We start by showing $\pref(\phi)$ is $\omega$-regular for $\ltlf$ formula $\phi$:

\begin{theorem}[\rm\bf Prefix automata: Upper bound]
\label{thrm:PropertiesPref} 
For an $\ltlf$ formula $\phi$,  the language $\pref(\phi)$ is $\omega$-regular.     
The B\"uchi automaton recognizing $\pref(\phi)$ has $2^{2^{\mathcal{O}(|\phi|)}}$ states. 
\end{theorem}

 \begin{proof}  
Given $\ltlf$ formula $\phi$, we construct a DBA for $\pref(\phi)$ as follows:
\begin{enumerate}
    \item Construct a DFA $D = (\alphabet, Q, \init, \delta, F)$ for $\neg\phi$, i.e.,  $\mathcal{L}(D) = \mathcal{L}(\neg\phi)$.
    
    We require $D$ to be \emph{complete} in the sense that for every state $s$ and every alphabet $a \in \alphabet$, there exists a successor $t = \delta(s, a)$.
    
    \item Obtain a DBA $C = (\alphabet, Q, \init, \delta', F)$ by converting all accepting states $F$ of $D$ to accepting sink states in $C$. For this, replace all outgoing transitions from all accepting states in $D$ with self loops on all letters.
    
    Formally, replace every $\delta(f,a) = t$ in DFA $D$ with $f = \delta'(f, a)$ in DBA $C$, for all $f\in F$ and $a \in \alphabet$. For all other states, let $\delta'$ behaves identically to $\delta$. 
    
    \item Obtain the desired B\"uchi automaton $B = (\alphabet, Q, \init, \delta',  \Final = Q\setminus F)$ by swapping accepting and non-accepting states of $C$.
\end{enumerate}
Since $C$ is a DBA with accepting sink states, $C$ is the complement of $B$. Hence, it suffices to show that $C$ accepts $w \in \Sigma^\omega$ iff there exists a finite prefix of $w$ that satisfies $\neg\phi$. 
Clearly,  $w \in \lang{C}$ then $w$ must have a finite-prefix satisfying $\neg\phi$ since the accepting states of $C$ and $D$ are identical.
Conversely, we need to show that despite $\delta$ and $\delta'$ being different, $C$ will accept all words that contain a finite prefix satisfying $\neg\phi$. For this, we show that for every such word, $C$ retains the transitions to accept the shortest prefix satisfying $\neg\phi$. Details can be found in the appendix. 
Finally, the number of states of $C$ are bounded by those of $D$ which is doubly exponential in $|\phi|$~\cite{de2013linear}.
\qed
\end{proof}

Observe that the B\"uchi automaton $B$ constructed above is \emph{deterministic}. One of our key discoveries is that the doubly exponential blow-up appears even in the construction of NBAs for $\pref(\phi)$, demonstrating that the blow-up is fundamentally unavoidable.
Theorem~\ref{thrm:prefixdouble} presents such an $\ltlf$ formula to demonstrate the blow-up. The rest of the section builds up to that construction.

We observe that the blow-up is caused by the combination of two aspects: First is the universal quantification on prefixes of words in $\pref(\phi)$; Second is the ability of an $\ltlf$ formula to identify the $k$-th last positions of finite words using the $\ltlfX$ (Next) modality. At first, we identify an $\omega$-regular language, parameterized with $n \geq 1$, such that all NBAs accepting the language have at least $2^{2^{n}}$ states.
Let $n \in \N$ and $\Sigma = \{0,1,\#, \&\}$. 
Consider the language $L_n \subseteq \Sigma^ \omega$ where
\[u\cdot \& \cdot v \in L_n \text{ s.t. if }   \#w\#  \text{ appears in } v \text{ then } \#w\# \text{ also appears in } u,\]
where $w\in\{0,1\}^n$, $u \in \{0,1,\#\}^*$ and $v \in \{0,1,\#\}^\omega$.
Intuitively, $L_n$ consists of infinite words that are (a) split into two parts by a special character ``$\&$" and (b) all words of the form $\# w\#$ appearing after ``$\&$" must have appeared before ``$\&$", for all $n$-length words $w \in \{0,1\}^n$.
Essentially, $L_n$ is  a bit-level adaption of the language $K_d$ where $x \cdot \& \cdot y \in K_d$ if digits appearing in $y$ are a subset of digits appearing in $x$, where $x \in D^*$ and  $y \in D^\omega$ for $D = \{0,1,\cdots, d-1\}$.
Obviously, the words $14\&1$ and $134\&4$ are good prefixes of a word $x\cdot \& \cdot y \in K_d$ when $d > 5$.
There are also less obvious good prefixes, such as a permutation of $D$ followed by the letter $\&$.
We need to recognize all good prefixes in order to accept the language $K_d$.
So, it is necessary to keep track of the digits (i.e., subsets of $D$) that the automaton has seen so far in an input word.
Hence, the NBA of $K_d$ needs $2^{\Omega(d)}$ states.
The same proof can be adapted to show that the NBA of $L_n$ consists of $2^{2^{\Omega(n)}}$ states. 
We defer a full proof to the supplemental material. 

Next, we need to identify a regular language $F_n$ such that, by abuse of notation, $\pref(F_n)$  corresponds to $L_n$ and $F_n$ can be represented by an $\ltlf$ formula of polynomial length in the parameter $n> 0$. 
A natural choice would be to let $F_n$ to be the finite-word version of $L_n$. In other words, $u\cdot \& \cdot v \in F_n$ s.t. if $\#w\#$ appears in $v$ then $\#w\#$ must have appeared in $u$ for all $w \in \{0,1\}^n$ and $u,v \in \{0,1,\#\}^*$. The issue is that $F_n$ cannot be represented by a short $\ltlf$ formula for the same reason why  $L_n$ cannot be expressed by a short $\ltl$ formula. 

We need $F_n$ to be  a {\em simpler} language. The roadmap would be to leverage the universal quantification over all prefixes to generate $L_n$. This is also where we leverage the ability of $\ltlf$ to refer to the last $k$-th positions of a finite trace. Keeping these goalposts, we define regular language $F_n\subseteq\Sigma^*$ as 
\begin{align*}
  u\cdot \& \cdot v \in F_n \text{ s.t. } &\text{if the last }n+2  \text{ characters of } v \text{ are of the form }  \#w\# \\ 
  & \text{ then } \#w\# \text{ also appears in } u,  
\end{align*}
where $w \in \{0,1\}^n$ and $u,v\in \{0,1,\#\}^*$.
Intuitively, by applying universal quantification on all finite-length prefixes, focusing on the last $n+2$ characters of words in $F_n$ is sufficient to ensure that every occurrence of the form $\#w\#$ after the symbol ``$\&$" appears in the portion before  the ``$\&$". 

There is one last caveat.
There are infinitely many prefixes of words in $L_n$ that may not contain the symbol $\&$. This issue can be easily remedied by including words without symbol $\&$ to both languages. We overload the notation of $\pref(L)$ to refer to the prefix language of a language over finite words $L$. Then,

\begin{lemma}
\label{lemma:prefixlangauge}
Let $L_n$ and $F_n$ be as defined above. Then $$L_n \uplus \{0,1,\#\}^\omega = \pref(F_n \uplus \{0,1,\#\}^*).$$
\end{lemma}
\begin{proof}[Proof Sketch]
To see why $L_n \uplus \{0,1,\#\}^\omega \subseteq \pref(F_n \uplus \{0,1,\#\}^*)$, observe that the prefixes of a word  $w \in  L_n \uplus \{0,1,\#\}^\omega$ either contain the symbol $\&$ or they don't. If the prefix falls under the latter, then the prefix is contained in $\{0,1,\#\}^*$. Otherwise, if the last $n+2$ characters are not in the form $\#w\# $ for $w \in \{0,1\}^n$ then the prefix is contained in $F_n$ by definition of $F_n$. If the last $n+2$ characters are in form $\#w\# $ for $w \in \{0,1\}^n$, then, by properties of words in $L_n$, $\#w\# $ must have appeared before $\&$. Once again, the prefix is contained in $F_n$. Thus, all prefixes of $w$ are contained in  $F_n \uplus \{0,1,\#\}^*$. 

The converse, i.e., $\pref(F_n \uplus \{0,1,\#\}^*) \subseteq L_n \uplus \{0,1,\#\}^\omega $, can be proven by a similar case-by-case analysis.
Details can be found in the appendix. \qed
\end{proof}

The last piece is to show that the language 
$F_n \uplus \{0,1,\#\}^*$ can be expressed using an $\ltlf$ formula $\phi_n$ of length polynomial in $n$, as shown below:

\begin{theorem}[\rm\bf Prefix automata: Lower bound]
\label{thrm:prefixdouble}
There exists an $\ltlf$ formula $\psi$ such that the number of states in all NBAs for $\pref(\psi)$ is $2^{2^{\Omega(\sqrt{|\psi|}\ \hspace{-0.4mm})}}$. 

\end{theorem}
\begin{proof}

Let $n \in \N\setminus \{0\}$ and $\Sigma = \{0,1,\#, \&\}$. Let $L_n$ and $F_n$ be as defined above. 

Since all NBAs of $L_n$ are of size $2^{2^{\Omega{(n)}}}$ and $L_n$ is disjoint from $\{0,1,\#\}^\omega$ by containing the ``\&" symbol, it is easy to show that all NBAs of $L_n \uplus \{0,1,\#\}^\omega$ require  $2^{2^{\Omega{(n)}}}$  states as well. 

From Lemma~\ref{lemma:prefixlangauge}, it is sufficient to show that $F_n\uplus \{0,1,\#\}^*$ can be represented by an $\ltlf$ formula of length $\bigO(n^2)$. 
So, let us construct the desired $\ltlf$ formula $\phi_n$.
By abuse of notation, let the propositions be given by $\ap = \{0,1,\#,\&\}$ with the interpretation that the symbol holds when its proposition is true.
Recall that a letter $\sigma$ in the finite alphabet $\Sigma$ corresponds to a valuation over the atomic propositions $\ap$.
For instance, $\& \in \alphabet$ is interpreted as the valuation $\neg 0 \land \neg 1 \land \neg \# \land \&$ over $\ap$. 
Then, the $\ltlf$ formula $\phi_n$
 is a conjunction of the following three:

\begin{enumerate}
    \item[(R1).] At all times, only one proposition can be true.
    \item[(R2).] If ``$\&$" holds at some place, it occurs exactly once.
    \item[(R3).] If ``$\&$" holds at some  place, then if the end of the word has the form $\#w\#$, for $w \in \{0,1\}^n$, $\#w\#$ must have appeared before ``$\&$".
\end{enumerate}
The $\ltlf$ formulation of (R1), denoted by $\mathsf{OnlyOneProp}$, is quite straightforward and has been deferred to the supplementary material. 
The formulation of (R2) is $\ltlfF \& \rightarrow \mathsf{ExactOne\&}$, where  $\mathsf{ExactOne\&}$ expresses that $``\&"$ occurs exactly once: 
$$\mathsf{ExactOne\&}:=(\neg \& \ltlfU (\&   \wedge (\neg(\ltlfX \ltlftrue) \vee \ltlfX (\ltlfG \neg \&))) ).$$
Intuitively, the $``\&"$ symbol is not seen \emph{until} it is seen somewhere, after which either the trace  \emph{terminates} (i.e., $\neg(\ltlfX \ltlftrue)$ holds) or the trace does not see $``\&"$ \emph{globally} (i.e., $\ltlfX (\ltlfG \neg \&)$ holds).
In fact, we also have $\neg(\ltlfX \ltlftrue) \vee \ltlfX (\ltlfG \neg \&) \equiv \ltlfN (\ltlfG \neg \&)$.

To express (R3), we first introduce two formulas. The first is  $\mathsf{EndWith\#w\#}$ to express that the end of the word has the form $\#w\#$. The second is $\mathsf{End\#w\#AppearsBefore\&}$ to express that the word $\#w\#$ must appear before $``\&"$. So, (R3) is expressed by 
$$\ltlfF \&  \rightarrow (\mathsf{EndWith\#w\#} \rightarrow \mathsf{End\#w\#AppearsBefore\&})$$
For $\mathsf{EndWith\#w\#}$, we introduce shorthands, namely  $\mathsf{Ends} := {\ltlfX}^{n+1}(\neg(\ltlfX \ltlftrue))$, and $\mathsf{Appear\#w\#}:= \#  \land {\ltlfX}^{n+1}\# \land \bigwedge_{i=1}^{n} {\ltlfX}^i (0 \lor 1)$.
Note that $\mathsf{Ends}$ is true only at the $(n+2)$-th last position  of a trace and $\mathsf{Appear\#w\#}$ enforces that the current and next $n+1$ positions have the form $\#w\#$ for $w \in \{0,1\}^n$.
Then, $$\mathsf{EndWith\#w\#}:= \ltlfG(\mathsf{Ends} \rightarrow \mathsf{Appear\#w\#})$$

Also, 
$\mathsf{End\#w\#AppearsBefore\&}:=
$ \begin{align*}  \ltlfF\Big(\mathsf{Appear\#w\#} \land \ltlfF\&  \land \bigwedge_{i=1}^n[(\ltlfX^i 0 \land \ltlfG(\mathsf{Ends} \rightarrow \ltlfX^i 0)) \vee (\ltlfX^i 1 \land \ltlfG(\mathsf{Ends} \rightarrow \ltlfX^i 1))]
\Big)
\end{align*}
Intuitively, when defining $\mathsf{End\#w\#AppearsBefore\&}$, we assume that we are standing at the first position of a word of the form $\#w\#$ that appears before $``\&"$.
So, we require that $\mathsf{Appear\#w\#}$ holds and later $\ltlfF \&$ holds.
Next, we require the same word $w$ to appear at the end.
So we require that if in the $i$-th position,  $0$ (resp. $1$) holds, at the $i$-th position from where $\mathsf{Ends}$ holds,  $0 $ (resp. $1$) must also hold. 
This is formulated as $(\ltlfX^i 0 \land \ltlfG(\mathsf{Ends} \rightarrow \ltlfX^i 0)) \vee (\ltlfX^i 1 \land \ltlfG(\mathsf{Ends} \rightarrow \ltlfX^i 1))$.

Finally, the whole formula $\phi_n$ is given as follows:
\begin{align*}
    \phi_n  & = \mathsf{OnlyOneProp} \\
        & \land (\ltlfF\& \rightarrow (\mathsf{ExactOne\&} \land  ((\mathsf{EndWith\#w\#} \rightarrow \mathsf{End\#w\#AppearsBefore\&}))))
\end{align*}

Clearly, when $\ltlfF \&$ does not hold, all words satisfying $\phi_n$ would be in $\{0,1,\#\}^{\omega}$.
If $\ltlfF \&$ holds, then all words should meet (R2) and (R3).
One can easily verify that $\phi_n$ specifies the language $F_{n} \uplus \{0,1,\# \}^*$. Thus, $\pref(\phi_n) = L_n \uplus \{0,1,\# \}^\omega$.

Last but not the least, the length of $\phi_n$ is in $\bigO(n^2)$ since $\mathsf{End\#w\#AppearsBefore\&}$ has length of $\bigO(n^2)$. 
\qed
\end{proof}

Note that the $\ltlf$ formulation makes heavy use of $\mathsf{Ends}$, which in turn uses the $\ltlfX$ modality. Essentially, $\mathsf{Ends}$ serves as a unique identifier of a specific position at the end of all traces. This enables us to anchor at that location without any artificial constructs and to express the desiderata accordingly. This is a crucial difference between  $\ltlf$ and $\ltl$.

\subsection{Prefix automata for $\ltlf$ Fragment}
\label{ssec:prefix-fragment}
In this section, we show that a singly exponential construction of NBAs is possible for a fragment of $\ltlf$ formulas.
Through an exposition of the prefix language for fragments of $\ltlf$, we highlight some of the peculiarities of the prefix language.
Consider the fragment of $\ltlf$, denoted as $\ltlf_{\setminus \setnocond{\ltlfR, \lor}}$, 
 which permits all but the $\ltlfR$ (Release) modality and allows $\neg$ and $\vee$ on literals only, as defined below:
$$\psi := \ell \mid \neg \ell \mid  \psi \land \psi \mid \ltlfX \psi  \mid \ltlfN \psi \mid  \ltlfF \psi \mid \ltlfG \psi \mid \psi \ltlfU \psi$$ where $\ell := a \in \ap \mid \neg a \mid \ell \land \ell \mid \ell \lor \ell$. 
We show that the prefix language of this fragment is equivalently represented by an $\ltl$ formula of the same size, hence its NBA is singly exponential in the size of the formula. The said $\ltl$ formula can be obtained using the translation
$t: \ltlf_{\setminus \setnocond{\ltlfR, \lor}} \rightarrow \ltl$  described below (Since $\ltl$ and $\ltlf $ share  the same syntax, to  avoid confusion, we add the subscript $\infty$ to temporal operators in $\ltl$, indicating that we have $|\rho| = \infty$. For instance, Globally in $\ltl$ becomes $\ltlfG_\infty$):


\begin{multicols}{2}
    \begin{itemize}
        \item $t(\ell) = \ell$, $t(\neg \ell) = \neg \ell$ 
        \item $t(\ltlfX \psi) = \ltlffalse$, $t(\ltlfN \psi) = \ltlfX_{\infty} t(\psi)$
        \item $t(\psi_1 \land \psi_2) = t(\psi_1) \land t(\psi_2)$
        \item $t(\ltlfF \psi) = t(\psi) $
        \item $t(\psi_1 \ltlfU \psi_2) = t(\psi_2)$
        \item $t(\ltlfG \psi) = \ltlG (t(\psi))$ 
    \end{itemize}
\end{multicols}

The insight behind this translation is to identify that the criteria for a formula to hold on all finite-length prefixes simplifies to the formula holding on a prefix of length one. The proof is presented below:

\begin{lemma}
    \label{lem:ltlf-f-g}
    Let $\phi \in \ltlf_{\setminus \setnocond{\ltlfR, \lor}}$ and let $\ltl$ $t(\phi)$ be  as defined above.
    Then, $\lang{t(\phi)} = \pref(\phi)$ and $\mathcal{O}(|\phi|) = \mathcal{O}(|t(\phi)|)$.
\end{lemma}
\begin{proof}
    Trivially, $\mathcal{O}(|\phi|) = \mathcal{O}(|t(\phi|)$ holds. 
    We prove that $\lang{t(\phi)} = \pref(\phi)$ by structural induction on $\phi$. In the interest of space, we skip the base cases ($\ell$ and $\neg\ell$). We also skip the $\land$ and $\ltlfG$ modalities, as they are intuitive. We present the argument for $\ltlfX$, $\ltlfN$, $\ltlfF$, and $\ltlfU$. The full proof has been deferred to the appendix.

    We set up notations: for $w = w_0 w_1 \cdots\in \alphabet^{\omega}$, let $w[i,j]=w_i\cdots w_{j-1}$ denote subsequences of $w$ for $0\leq i<j$.
    So, $w[0,n]$ is the $n$-length prefix of $w$ for $n > 0$.
     By inductive hypothesis (I.H.), we assume $\lang{t(\gamma)} = \pref(\gamma)$ for $\gamma \in \{\psi, \psi_1, \psi_2\}$.   
     
     \begin{description}
        \item[Case $\ltlfF \psi$:]
        The critical observation is that for $\ltlfF\psi$ to hold on all finite prefixes,  $\ltlfF\psi$ must hold on the prefix of length 1, which in turn is possible only if the first position of the word satisfies $\psi$. Formally,  first  we show that $\pref(\ltlfF\psi) \subseteq\lang{t(\ltlfF\psi)}$. Let $w \in \pref(\ltlfF\psi)$. Then, in particular $w[0,1] \models \ltlfF\psi$. This is possible only if $w[0,1] \models \psi$. Thus, for all $n>0$, we get $w[0,n] \models \psi$. So, $w \in \pref(\psi)$. By I.H., $w \in \lang{t(\psi)}$. By translation, this means $w \in \lang{t(\ltlfF\psi)}$.
        Next, we show $\lang{t(\ltlfF\psi)} \subseteq \pref(\ltlfF\psi)$. Let $w \in \lang{t(\ltlfF\psi)}$. By translation, $w \in \lang{t(\psi)}$. By I.H., $w \in \pref(\psi) $. Now, if $\psi$ holds, then $\ltlfF\psi$ also holds for all non-zero lengths. Hence, $w \in \pref(\ltlfF\psi)$.

        \item[Case $\psi_1\ltlfU\psi_2$:]
        As earlier, the critical observation is for $\psi_1\ltlfU\psi_2$ to hold on a prefix of length one. For this, $\psi_2$ must hold. The proof is similar to the earlier case.   

        \item[Case $\ltlfX \psi$:] The issue is that $\ltlfX \psi$ can never be true on a word of length one, since there does not exist a next position on length one words. Hence, $\pref(\ltlfX \psi) = \emptyset = \lang{\mathsf{False}} = \lang{t(\ltlfX \psi)}$. 

        \item[Case $\ltlfN \psi$:] $\ltlfN$ (Weak Next) doesn't have the issue faced by $\ltlfX$. If a word is of length one, $\ltlfN \psi$ trivially holds. For words of all other lengths, it requires $\ltlfX \psi$ to hold. Formally, first we show that $\pref(\ltlfN \psi) \subseteq \lang{t(\ltlfN \psi)}$. Let $w \in \pref(\ltlfN \psi) $.  Then, by semantics of $\ltlf$, it follows that the second position on $w$ must satisfy $\psi$, i.e., $w[1,2] \models \psi$. In particular, for all $i>1$, $w[1,i] \models \psi$. So, $w[1,\infty] \in \pref(\psi)$.
        By I.H., $w[1,\infty] \in \lang{t(\psi)}$. Hence, $w \in \lang{\ltlfX_\infty t(\psi)} = \lang{t(\ltlfN\psi)}$.
        Conversely, let $w \in \lang{t(\ltlfN\psi)}$. By translation, $w \in \lang{\ltlfX_\infty t(\psi)}$.
        Hence, by I.H., we get for all $i>1$, $w[0,i] \models \ltlfX \psi$ and $w[0,1] \models \ltlfN\psi$ since $w[1, \infty] \in \lang{t(\psi)} = \pref(\psi) $. In other words, $w \in \pref(\ltlfN\psi)$.
    \end{description}
    \qed
\end{proof}

An immediate consequence of Lemma~\ref{lem:ltlf-f-g} is that the prefix automata for $\ltlf_{\setminus \setnocond{\ltlfR, \lor}}$ are singly exponential in the size of the formula~\cite{DBLP:journals/iandc/VardiW94}:
\begin{corollary}
\label{thm:ltlf-f-g-prefix}
    Let $\phi \in \ltlf_{\setminus \setnocond{\ltlfR, \lor}}$.
    The NBA for $\pref(\phi)$ contains $2^{\bigO(\size{\phi})}$ states.
\end{corollary}

Note that, in all the cases above, every conjunct holds on \emph{all} finite prefixes. This may not be true if  $\lor$ (or) is permitted in the formula. For example, consider $\phi = \ltlfG a \lor \ltlfF b$. Now, the word $w = \{a\}\{b\}\{\}^{\omega} \in \pref(\phi)$ since the prefix of length one satisfies $\ltlfG a$ and all other prefixes satisfy $\ltlfF b$. Hence, with disjunction, different prefixes can satisfy \emph{different} disjuncts. In fact, the $\ltl$ formula for $\pref(\phi)$ is $a \ltlfU_\infty b \lor \ltlfG_\infty a$. However, such translations may increase the formula length because of duplicating the formula under $\ltlG$ modality. An open problem here is to identify the largest fragment for which the prefix automata have only singly exponential blow-up. This goes hand-in-hand with uncovering the core behind the doubly exponential blow-up for prefix automata.

\section{Complexity of $\ltlf$ Model Checking}
\label{sec:complexity}

We present the complexity of $\ltlf$ model checking.
Section~\ref{sec:lowerbound} develops the lower bound for model checking non-terminating systems and  Section~\ref{sec:final} presents the completeness argument for both terminating and non-terminating systems.


\subsection{\expspace \ Lower Bound for Non-terminating Systems}
\label{sec:lowerbound}

\sloppy 

We prove \expspace-hardness of $\ltlf$ model checking of non-terminating systems by a polynomial-time reduction from the problem of whether an exponential-space Turing machine $T = (Q, \Gamma, \delta, q_0, F)$ accepts an input word $x = x_1 \ldots x_n$. The components of the Turing machine are defined as follows:

\begin{itemize}
    \item $Q$ is the set of states and  $q_0 \in Q$ is the initial state.
    \item $\Gamma$ is the tape alphabet, which is assumed to include the blank symbol $\emptyset$.
    \item $\delta : Q \times \Gamma \rightarrow Q \times \Gamma \times \{\leftarrow, \rightarrow\}$ is the transition function. $\delta(q, \gamma) = (q', \gamma', d)$ means that if the machine is in state $q$ and the head reads symbol $\gamma$, it moves to state $q'$, writes symbol $\gamma'$, and moves the head in direction $d$. 
    \item $F \subseteq Q$ is the set of accepting states. The machine accepts if it reaches a state in $F$.
\end{itemize}
Since $T$ is an exponential-space Turing machine, we can assume that its tape has $2^{cn}$ cells, where $n$ is the size of the input and $c$ is a constant.

\subsubsection{High-Level Idea}

Given a Turing machine $T$ and an input $x$, our reduction will construct a non-terminating system $M$ and an $\ltlf$ formula $\varphi$ s.t. $T$ accepts $x$ iff every execution of $M$ has a finite prefix that satisfies $\varphi$, i.e., $M \models \varphi$.


In this reduction, we will encode runs of the Turing machine as label sequences of the system. A {\em cell} in the tape is encoded as a sequence of $cn + 1$ propositional assignments. The first assignment encodes the content of the cell, which can be either a symbol $\gamma \in \Gamma$ or a symbol $\gamma$ along with a state $q \in Q$, the latter indicating that the head is on that cell and is in state $q$. The remaining $cn$ assignments encode the position of the cell in the tape as a $cn$-bit number (since the tape has $2^{cn}$ cells).
The concatenation of $2^{cn}$ cells encodes a {\em configuration} of the Turing machine. Therefore, each configuration is encoded by $2^{cn} (cn + 1)$ assignments in total. The concatenation of configurations encodes a {\em run} of the Turing machine. Note, however, that for such a run to be consistent with the run of $T$ on $x$, certain consistency conditions must hold:

\begin{enumerate}
    \item For every configuration, the encoding of the position of the first cell must be $0$, and the encoding must increase by $1$ for each successive cell.
    \item The first configuration must start with $x$ on the tape and the head on the first cell and in the initial state $q_0$.
    \item Successive configurations must be consistent with the transition function $\delta$.
\end{enumerate}

One way is to enforce all consistency conditions through the system $M$. However, since each configuration consists of $2^{cn}$ cells, this would require the system to have an exponential number of states. Therefore, to allow for a polynomial reduction, we enforce the consistency conditions through the formula $\varphi$.

For this, we construct an $\ltlf$ formula $\varphi := \varphi_{cons}  \rightarrow \varphi_{acc} $.
where $\varphi_{cons}$ expresses the the consistency conditions and $\varphi_{acc}$ expresses the property of reaching an accepting configuration. Therefore, every execution with a finite prefix that satisfies $\varphi$ is either inconsistent or an accepting run of $T$ on $x$. Since $T$ is deterministic, there is exactly one execution of $M$ that is consistent with $T$. Every other execution will necessarily satisfy $\neg \varphi_{cons}$, and this execution will satisfy $\varphi_{acc}$ if and only if $T$ accepts $x$. Therefore, if every execution of $M$ has a finite prefix that satisfies $\varphi$, then the run of $T$ on input $x$ is accepting, and vice-versa.


We now provide the details of the system $M$ and the formula $\varphi$.

\subsubsection{Atomic Propositions} \label{sec:prop}

The propositions used by system $M$ are the following:

\begin{itemize}
    \item $part_0$ indicates that the current assignment represents the first part of the cell encoding, encoding the cell's content.
    \item $part_i$, for $1 \leq i \leq cn$, indicates that the current assignment represents the $i$-th bit of the encoding of the cell's position. Only one of $part_0, \ldots, part_{cn}$ is true at any given time.
    \item $cell_\lambda$, for $\lambda \in \Gamma \cup (Q \times \Gamma)$, indicates that the content of the cell is $\lambda$ (a tape symbol with or without the head). This proposition can only be true if $part_0$ is true.
    \item $bit$ gives the current bit of the cell's position. This proposition can only be true if $part_0$ is false.
\end{itemize}

\subsubsection{The Model} \label{sec:model}

We define the transition system $M = (\Sigma, S, T, \iota, L)$ as follows:

\begin{itemize}
    \item $\Sigma = \{part_0, \ldots, part_{cn}\} \cup \{cell_\lambda \mid \lambda \in \Gamma \cup (Q \times \Gamma)\} \cup \{bit\}$
    \item $S = \{(0, \lambda) \mid \lambda \in \Gamma \cup (Q \times \Gamma)\} \cup \{(i, b) \mid 1 \leq i \leq cn, b \in \{0, 1\}\}$
    \item $\iota = (0, (q_0, \emptyset))$
    \item $(s, s') \in T$ if and only if one of the following is true (for some $\lambda, b, b'$):
    
    \begin{itemize}
        \item $s = (0, \lambda)$ and $s' = (1, b)$.
        \item $s = (i, b)$ for $1 \leq i < cn$, and $s' = (i + 1, b')$.
        \item $s = (cn, b)$ and $s' = (0, \lambda)$.
    \end{itemize}
    
    \item $L((0, \lambda)) = \{part_0, cell_\lambda\}$
    \item $L((i, b)) = \{part_i\} \cup \{bit \mid b = 1\}$

    
    
\end{itemize}

The propositional alphabet $\Sigma$ consists of the set of propositions described above. The states of the $M$ are either of the form $(0, \lambda)$, where $\lambda$ is the content of a cell, or $(i, b)$ for $1 \leq i \leq cn$, where $b$ is the current bit in the encoding of the cell's position. The initial state is $(0, (q_0, \emptyset))$, indicating that a) this is the first part of the cell's encoding, b) the head is on this cell, c) the machine is in the initial state $q_0$, and d) the cell is blank (this should be the cell immediately to the left of the input word $x$).

The transition relation ensures only that the system progresses consistently from part 0 of the encoding to part 1, part 2, part 3, and so on until part $cn$, after which it resets back to part 0 (of the next cell). Note that the values of $\lambda$ and $b$ are unconstrained, as these will be handled by the formula $\varphi$. Observe the three consistency conditions required for runs of $T$ are not wired into the model. 

Finally, the labeling function $L$ simply converts the state into an appropriate propositional representation.

\subsubsection{The Formula}

We now construct the $\ltlf$ formula $\varphi$ over the propositional alphabet $\Sigma$. As mentioned before, we want $\varphi$ to be such that, if an execution of the system $M$ has a prefix that satisfies $\varphi$, then either that execution violates a consistency condition or it is an accepting run.
To achieve this, we construct $\varphi = \neg \varphi_{cons} \lor \varphi_{acc}$. 
$\varphi_{acc}$ is defined as follows: $$\varphi_{acc} = \bigvee_{q \in F} \bigvee_{\gamma \in \Gamma} \ltlfF\,cell_{(q, \gamma)}.$$ It is easy to see that an execution of $M$ has a prefix that satisfies $\varphi_{acc}$ iff that execution reaches a state $(0, (q, \gamma))$ where $q$ is an accepting state of $T$.

Meanwhile, we define $\varphi_{cons}$ as a conjunction of formulas, such that if an execution has a prefix that violates one of these formulas then the execution is inconsistent, and every inconsistent execution has a prefix that violates one of these formulas. We classify these formulas into three groups, one for each of the three consistency conditions described above:

\begin{enumerate}
    \item[(C1).] Consistency within a configuration (the binary encoding of each cell's position is correct)
    \item[(C2).] Consistency with the input word (the first configuration is correct)
    \item[(C3).] Consistency with the transition function (every configuration follows from the previous one)
\end{enumerate}

The first two conditions (C1) and (C2) are relatively straightforward to encode as formulas of polynomial size. For details, refer to the appendix. 

The third condition (C3) is where the biggest challenge lies. 
This condition requires reasoning about changes from one configuration to the next. The difficulty lies in accessing the segment that represents the same cell in the next configuration using a polynomial-sized formula. Recall that a cell is represented by $cn + 1$ assignments in the trace and each configuration is composed of $2^{cn}$ cells.  Since the size of each configuration is exponential, formulas may require exponential size. For instance, if the segment representing a cell begins at assignment $i$ in the trace, then the same cell in the next configuration will start at assignment $i + 2^{cn}(cn + 1)$. Referring to this assignment directly in the formula would require $2^{cn}(cn + 1)$ nested $\ltlfX$ operators. Alternatively, the cell in the next configuration can  be identified by being the first cell where the binary encoding of its position on the tape is the same as the current cell. However, this may require enumeration on all possible assignments of the $cn + 1$ bits.

To circumvent this problem and compare corresponding cells in two different configurations using a formula of polynomial size, we take advantage of the fact that we are dealing with finite prefixes of the trace. The insight is that we can use the last position in the trace as an anchor, so that instead of having to find the cell in the next configuration with the same position encoding, we can instead look at the last cell in the trace and test if a) it is in the next configuration, and b) it has the same position encoding. Since the formula is checked for every prefix, eventually we will find a prefix where this holds. We can then check if the contents of the cells are consistent with the transition function.

We now go into details of the formula for (C3).
Consistency condition (C3) says that every configuration follows from the previous one according to $T$'s transition function $\delta$. As mentioned before, to ensure that we get a formula of polynomial size, the formula that we construct actually expresses the following condition: for all cells $c$ in the prefix, if the last cell $c_{Last}$ of the prefix is in the same position as $c$ but in the next configuration, then $c_{Last}$ follows from $c$ based on the transition function. Since the formula must hold for all prefixes, its satisfaction implies the original consistency condition.

We start by defining the useful shorthand $L^{-i} \phi \equiv \ltlfF(\phi \land \ltlfX^{i-1} \neg \ltlfX\,\ltlftrue)$, which denotes that $\phi$ holds $i$ positions before the end of the prefix (e.g. $L^{-1}\phi$ means that $\phi$ holds at the last position of the prefix). This is expressed by saying that at some point in the future $\phi$ holds, and $i - 1$ positions after that is the last position of the prefix (by the semantics of $\ltlf{}$, $\neg \ltlfX\,\ltlftrue$ only holds at the last position).
We then define the formula $\mathsf{MatchLastCell}$, which checks if the cell $c$ in the current position corresponds to the last cell $c_{Last}$ of the prefix, as follows:
\begin{align*}
    \mathsf{MatchLastCell} \equiv & \text{  } part_0 \land L^{-cn} part_0 \land \bigwedge^{cn}_{i = 1} (\ltlfX^i bit \leftrightarrow L^{-cn} \ltlfX^i bit) \\
        \land & \text{  } \ltlfX \Big(\neg \mathsf{NewConfig}\,\ltlfU\,\big(\mathsf{NewConfig}   \land \ltlfX\,\ltlfG\,\neg \mathsf{NewConfig}\big)\Big)
\end{align*}
where $\mathsf{NewConfig} \equiv (part_0 \land \bigwedge^{cn}_{i = 1} (\ltlfX^i \neg bit))$ denotes the start of a new configuration (a cell whose position in the tape is encoded as $0$). $\mathsf{MatchLastCell}$ expresses that (a) we are at the start of a cell $c$ ($part_0$); (b) the last $cn$ positions of the prefix encode another cell $c_{Last}$ ($L^{-cn} part_0$); (c) $c$ and $c_{Last}$ are in the same tape position ($\bigwedge^{cn}_{i = 1} (\ltlfX^i bit \leftrightarrow L^{-cn} \ltlfX^i bit)$); and (d) we start a new configuration exactly once between $c$ and $c_{Last}$ ($\ltlfX(\neg \mathsf{NewConfig}\,\ltlfU\,(\mathsf{NewConfig} \land \ltlfX\,\ltlfG\,\neg \mathsf{NewConfig}))$). In other words, $c$ and $c_{Last}$ are the same cell in successive configurations.
We can then encode the consistency condition by the formula
\begin{align*}
    & \ltlfG(\mathsf{MatchLastCell} \rightarrow \varphi_\delta)\,\land 
    \ltlfG(\mathsf{MatchLastCell} \rightarrow \varphi^{\leftarrow}_\delta)\,\\
    \land \  &\ltlfG(\ltlfX^{cn+1}\,\mathsf{MatchLastCell} \rightarrow \varphi^{\rightarrow}_\delta)\,\land 
    \ltlfG(\ltlfX^{cn+1}\,\mathsf{MatchLastCell} \rightarrow \varphi^0_\delta)
\end{align*}
where each of $\varphi_\delta$,  $\varphi^{\leftarrow}_\delta$,  $\varphi^{\rightarrow}_\delta$, and $\varphi^0_\delta$ expresses one way in which the contents of the cell $c$ can change (or not change) in the next configuration:


\begin{itemize}
    \item $\varphi_\delta$ expresses that if the head is on $c$ ($cell_{(q, \gamma)}$), then in $c_{Last}$ the head must have moved to a different cell and written the appropriate symbol $\gamma'$ given by the transition relation ($L^{-cn}\,cell_{\gamma'}$)
    
    \item $\varphi^{\leftarrow}_\delta$ expresses that if the head is on the cell to the \emph{right} of $c$ ($\ltlfX^{cn + 1}\,cell_{(q, \gamma_2)}$), and the transition relation requires it to move left, then in the next configuration the head must have moved to $c_{Last}$ ($L^{-cn}\,cell_{(q', \gamma_1)})$)

    \item $\varphi^{\rightarrow}_\delta$ expresses that if the head is on the cell to the \emph{left} of $c$ ($cell_{(q, \gamma_1)}$), and the transition relation requires it to move right, then in the next configuration the head must have moved to $c_{Last}$ ($L^{-cn}\,cell_{(q', \gamma_2)})$)

    \item Finally, $\varphi^0_\delta$ expresses that if the head is neither on $c$ nor on the cells adjacent to it ($cell_{\gamma_1} \land \ltlfX^{cn + 1}\,cell_{\gamma_2} \land \ltlfX^{2(cn + 1)}\,cell_{\gamma_3}$), then the contents of the cell don't change ($L^{-cn}\,cell_{\gamma_2}$)
\end{itemize}

Note that in the latter two formulas $c$ is the cell to the right of the current cell ($\ltlfX^{cn+1}\,\mathsf{MatchLastCell}$) this is necessary so that $\varphi^{\rightarrow}_\delta$ and $\varphi^0_\delta$ can refer to the cell to the left of $c$. Formula for $\varphi_\delta$,  $\varphi^{\leftarrow}_\delta$,  $\varphi^{\rightarrow}_\delta$, and $\varphi^0_\delta$ have been presented in the appendix. The size of each formula is polynomial in the size of the transition relation of the Turing Machine.

\begin{theorem}[$\ltlf$ Model Checking. Lower bound]
\label{thrm:expspacehardreduction}
$\ltlf$ model checking of non-terminating systems is \expspace-hard.  
\end{theorem}

\begin{proof}
Let the non-terminating system $M$ and $\ltlf$ formula $\varphi = \neg \varphi_{cons} \lor \varphi_{acc}$ be as described above. We show that an exponential-space Turing machine $T$ accepts an input word $x$ iff every execution of $M$ has a finite prefix that satisfies $\varphi$, i.e., $M\models \varphi$. 
Note that since $T$ is deterministic, its execution on the input word $x$ is unique. Therefore, there is exactly one trace $\pi$ of $M$ that simulates the execution of $T$ on $x$. By construction, a trace has a finite prefix that satisfies $\neg \varphi_{cons}$ iff that trace violates one of the consistency conditions. This holds for every trace of $M$ except $\pi$. 
So,  because no finite prefix of $\pi$ satisfies $\neg \varphi_{cons}$, 
$M$ model checks if and only if $\pi$ has a prefix that satisfies $\varphi_{acc}$, which means that $\pi$ eventually reaches an accepting state. Since $\pi$ simulates $T$ on $x$, this happens if and only if $T$ accepts $x$.
\qed
\end{proof}


\subsection{Final Complexity Results}
\label{sec:final}

Finally, we present the complexity of model-checking non-terminating systems:

\begin{theorem}[MC. Non-terminating. Complexity]
\label{thrm:upperboundnonprob}
$\ltlf$ model checking of non-terminating systems is \expspace-complete.
\end{theorem}

\begin{proof}
Recall, a non-terminating system $\M$ satisfies an $\ltlf$ formula  $\phi$ iff $\lang{\M}\cap \pref(\neg\phi) = \emptyset$.
A naive algorithm would explicitly construct $\pref(\neg\phi)$ and require doubly exponential space in the size of $\phi$. 
Instead, the approach is to construct $\pref(\phi)$ on-the-fly in exponential space and simultaneously evaluate the emptiness of $\M\cap \pref(\neg\phi)$. 
Given all three steps in the construction of $\pref(\phi)$ are amenable to on-the-fly constructions, this procedure follows standard on-the-fly procedures~\cite{vardi1986automata}. 
Thus, $\ltlf$ model checking of non-terminating models is in \expspace.
\autoref{thrm:expspacehardreduction} establishes the matching lower bound.
\qed
\end{proof}

This result is unexpected as it implies that $\ltlf $ model checking is exponentially harder than $\ltl$ model checking for non-terminating systems, contrary to the prior perception that problems in $\ltlf$ tend to be as hard if not easier than their counterparts in $\ltl$ (See~\autoref{tab:complexity}). 

Next, we present the complexity of  model-checking terminating systems:

\begin{theorem}[MC. Terminating. Complexity]
\label{thrm:completeterminating}
$\ltlf$ model checking of terminating systems is \pspace-complete.  
\end{theorem}
\begin{proof}
Recall that a terminating system $M$ satisfies an $\ltlf$ formula $\phi$ if every execution of $M$ satisfies $\phi$.
So, $M\models \phi$ iff $\lang{M\cap A_{\neg\phi}} = \emptyset$ where $A_{\neg\phi}$ is the NFA for $\neg\phi$. Since the NFA is exponential in the size of the $\ltlf$ formula~\cite{de2013linear}, an on-the-fly algorithm for non-emptiness checking of $M\cap A_{\neg\phi}$ can be performed in  \pspace. 
\pspace-hardness can be proven by a trivial reduction from {\em $\ltlf$ satisfiability}, which is  \pspace-complete~\cite{de2013linear}. 
\qed
\end{proof}

For $\ltlf $ synthesis, these results imply that it is much harder to verify a non-terminating transducer than a terminating transducer. Hence, to test the correctness of an $\ltlf$ synthesis tool by verifying its output strategy, it would be better for $\ltlf$ synthesis tools to generate terminating transducers. This, to the best of our knowledge, is the \emph{first} theoretically sound evidence to use one transducer over the other in $\ltlf$ synthesis.  
\section{Concluding Remarks}
\label{sec:discussion}

Motivated by the recent surge in $\ltlf$ synthesis tools that are rarely verified for result correctness, this work is the \emph{first} to investigate the problem of $\ltlf$ model checking. Noting that $\ltlf$ synthesis can generate both terminating and non-terminating transducers, we examine $\ltlf$ model checking for both possibilities. Our central result is that $\ltlf$ model checking of non-terminating models is exponentially harder than terminating models. Their complexities are $\expspace$-complete and $\pspace$-complete, respectively. 
This is surprising at first as it implies that $\ltlf $ model checking is harder than $\ltl$ model checking for non-terminating models, contrary to the expectation from prior comparisons between $\ltlf$ and $\ltl$ (See~\autoref{tab:complexity}). 
In addition to being of independent interest, our results immediately lend several broad impacts:
\begin{enumerate}
    \item They present the first theoretical evidence for the use of terminating transducers to represent the  synthesized strategies in $\ltlf$ synthesis, as it would be easier to verify the correctness of the synthesized transducer.
    \item Implementations of our $\ltlf$ model checking algorithms could be deployed in  large-scale competitions such as the $\ltlf$ track in SYNTCOMP 2023. 
    \item They invite further exploration into $\ltlf$ vs $\ltl$, as it breaks the prior perception that  problems in $\ltlf$ are  {\em as hard if not simpler} than their $\ltl$ counterparts. 
\end{enumerate}
    
These results inspire future work in the development of practical tools for model checking and synthesis as well as the development of $\ltlf$ model checking in more complex domains such as probabilistic models or under asynchrony~\cite{bansal2019synthesis,bansal2018synthesis}.  
It would be interesting to see how the practical implementations compare for $\ltlf$ model checking under terminating and non-terminating semantics, even though terminating models are preferred in theory.

\subsubsection*{Acknowledgements} We thank the anonymous reviewers for their valuable feedback.
This work has been supported by the Engineering and Physical Sciences Research Council [grant number EP/X021513/1], NASA 80NSSC17K0162, NSF grants IIS-1527668, CCF-1704883,
IIS-1830549, CNS-2016656, DoD MURI grant N00014-20-1-2787,
and an award from the Maryland Procurement Office.

%
%
\bibliographystyle{splncs04}
\bibliography{refs}

\begin{thebibliography}{10}
\providecommand{\url}[1]{\texttt{#1}}
\providecommand{\urlprefix}{URL }
\providecommand{\doi}[1]{https://doi.org/#1}

\bibitem{baier2006planning}
Baier, J.A., McIlraith, S.: Planning with temporally extended goals using
  heuristic search. In: ICAPS. pp. 342--345. AAAI Press (2006)

\bibitem{bansal2020hybrid}
Bansal, S., Li, Y., Tabajara, L., Vardi, M.: Hybrid compositional reasoning for
  reactive synthesis from finite-horizon specifications. In: AAAI. vol.~34, pp.
  9766--9774 (2020)

\bibitem{bansal2018synthesis}
Bansal, S., Namjoshi, K.S., Sa’ar, Y.: Synthesis of asynchronous reactive
  programs from temporal specifications. In: Computer Aided Verification: 30th
  International Conference, CAV 2018, Held as Part of the Federated Logic
  Conference, FloC 2018, Oxford, UK, July 14-17, 2018, Proceedings, Part I 30
  (2018)

\bibitem{bansal2019synthesis}
Bansal, S., Namjoshi, K.S., Sa'ar, Y.: Synthesis of coordination programs from
  linear temporal specifications. Proceedings of the ACM on Programming
  Languages (POPL) (2019)

\bibitem{blum1995designing}
Blum, M., Kannan, S.: Designing programs that check their work. Journal of the
  ACM  \textbf{42}(1),  269--291 (1995)

\bibitem{brafman2019planning}
Brafman, R.I., De~Giacomo, G.: Planning for {LTLf}/{LDLf} goals in
  non-markovian fully observable nondeterministic domains. In: IJCAI. pp.
  1602--1608 (2019)

\bibitem{camacho2019ltl}
Camacho, A., Icarte, R.T., Klassen, T.Q., Valenzano, R.A., McIlraith, S.A.:
  {LTL} and beyond: Formal languages for reward function specification in
  reinforcement learning. In: IJCAI. vol.~19, pp. 6065--6073 (2019)

\bibitem{de2021compositional}
De~Giacomo, G., Favorito, M.: Compositional approach to translate {LTLf}/{LDLf}
  into deterministic finite automata. In: Proceedings of the International
  Conference on Automated Planning and Scheduling. vol.~31, pp. 122--130 (2021)

\bibitem{de2022ltlf}
De~Giacomo, G., Favorito, M., Li, J., Vardi, M.Y., Xiao, S., Zhu, S.: Ltlf
  synthesis as and-or graph search: Knowledge compilation at work. In: Proc. of
  IJCAI (2022)

\bibitem{de2019foundations}
De~Giacomo, G., Iocchi, L., Favorito, M., Patrizi, F.: Foundations for
  restraining bolts: Reinforcement learning with {LTLf}/{LDLf} restraining
  specifications. In: ICAPS. vol.~29, pp. 128--136 (2019)

\bibitem{de2018automata}
De~Giacomo, G., Rubin, S.: Automata-theoretic foundations of fond planning for
  {LTLf} and {LDLf} goals. In: IJCAI. pp. 4729--4735 (2018)

\bibitem{de2015synthesis}
De~Giacomo, G., Vardi, M.: Synthesis for {LTL} and {LDL} on finite traces. In:
  IJCAI. pp. 1558--1564. AAAI Press (2015)

\bibitem{de2013linear}
De~Giacomo, G., Vardi, M.Y.: Linear temporal logic and linear dynamic logic on
  finite traces. In: IJCAI. pp. 854--860. AAAI Press (2013)

\bibitem{de2016ltlf}
De~Giacomo, G., Vardi, M.Y.: {LTLf} and {LDLf} synthesis under partial
  observability. In: IJCAI. vol.~2016, pp. 1044--1050 (2016)

\bibitem{duret2016spot}
Duret{-}Lutz, A., Renault, E., Colange, M., Renkin, F., Aisse, A.G.,
  Schlehuber{-}Caissier, P., Medioni, T., Martin, A., Dubois, J., Gillard, C.,
  Lauko, H.: From spot 2.0 to spot 2.10: What's new? In: Shoham, S., Vizel, Y.
  (eds.) Computer Aided Verification - 34th International Conference, {CAV}
  2022, Haifa, Israel, August 7-10, 2022, Proceedings, Part {II}. Lecture Notes
  in Computer Science, vol. 13372, pp. 174--187. Springer (2022)

\bibitem{esparza2020unified}
Esparza, J., K{\v{r}}et{\'\i}nsk{\`y}, J., Sickert, S.: A unified translation
  of linear temporal logic to $\omega$-automata. Journal of the ACM (JACM)
  \textbf{67}(6),  1--61 (2020)

\bibitem{favorito2023forward}
Favorito, M.: Forward ltlf synthesis: Dpll at work. arXiv preprint
  arXiv:2302.13825  (2023)

\bibitem{he2017reactive}
He, K., Lahijanian, M., Kavraki, L.E., Vardi, M.Y.: Reactive synthesis for
  finite tasks under resource constraints. In: IROS. pp. 5326--5332. IEEE
  (2017)

\bibitem{jacobs2023temporal}
Jacobs, S., Perez, G.A., Schlehuber-Caissier, P.: The temporal logic synthesis
  format {TLSF} v1.2 (2023)

\bibitem{kvretinsky2018owl}
K{\v{r}}et{\'\i}nsk{\`y}, J., Meggendorfer, T., Sickert, S.: Owl: a library for
  omega-words, automata, and {LTL}. In: ATVA. pp. 543--550. Springer (2018)

\bibitem{kuehlmann2002combinational}
Kuehlmann, A., van Eijk, C.A.: Combinational and sequential equivalence
  checking. Logic synthesis and Verification pp. 343--372 (2002)

\bibitem{DBLP:conf/litp/NicolaV90}
Nicola, R.D., Vaandrager, F.W.: Action versus state based logics for transition
  systems. In: Guessarian, I. (ed.) Semantics of Systems of Concurrent
  Processes, {LITP} Spring School on Theoretical Computer Science, La Roche
  Posay, France, April 23-27, 1990, Proceedings. Lecture Notes in Computer
  Science, vol.~469, pp. 407--419. Springer (1990)

\bibitem{pnueli1977temporal}
Pnueli, A.: The temporal logic of programs. In: FOCS. pp. 46--57. IEEE (1977)

\bibitem{pnueli1989synthesis}
Pnueli, A., Rosner, R.: On the synthesis of a reactive module. In: POPL. pp.
  179--190 (1989)

\bibitem{safra1988complexity}
Safra, S.: On the complexity of omega -automata. In: [FOCS. pp. 319--327 (1988)

\bibitem{siegel1998translation}
Siegel, M., Pnueli, A., Singerman, E.: Translation validation. In: Proc. of
  TACAS. pp. 151--166 (1998)

\bibitem{sistla1985complexity}
Sistla, A.P., Clarke, E.M.: The complexity of propositional linear temporal
  logics. Journal of the ACM (JACM)  \textbf{32}(3),  733--749 (1985)

\bibitem{tabajara2019partition}
Tabajara, L.M., Vardi, M.Y.: Partitioning techniques in {LTLf} synthesis. In:
  IJCAI. pp. 5599--5606. AAAI Press (2019)

\bibitem{TRV12}
Tabakov, D., Rozier, K., Vardi, M.Y.: Optimized temporal monitors for
  {SystemC}. Formal Methods in System Design  \textbf{41}(3),  236--268 (2012)

\bibitem{thomas2002automata}
Thomas, W., et~al.: Automata, logics, and infinite games: a guide to current
  research, vol.~2500. Springer Science \& Business Media (2002)

\bibitem{vardi2007buchi}
Vardi, M.Y.: The {B{\"u}chi} complementation saga. In: STACS. pp. 12--22.
  Springer (2007)

\bibitem{vardi1986automata}
Vardi, M.Y., Wolper, P.: An automata-theoretic approach to automatic program
  verification. In: LICS. IEEE Computer Society (1986)

\bibitem{DBLP:journals/iandc/VardiW94}
Vardi, M.Y., Wolper, P.: Reasoning about infinite computations. Inf. Comput.
  \textbf{115}(1),  1--37 (1994)

\bibitem{wells2020ltlf}
Wells, A.M., Lahijanian, M., Kavraki, L.E., Vardi, M.Y.: {LTLf} synthesis on
  probabilistic systems. arXiv preprint arXiv:2009.10883  (2020)

\bibitem{wolper1983reasoning}
Wolper, P., Vardi, M.Y., Sistla, A.P.: Reasoning about infinite computation
  paths. In: FOCS. pp. 185--194. IEEE (1983)

\bibitem{zhu2017ltlfsymbolic}
Zhu, S., Tabajara, L.M., Li, J., Pu, G., Vardi, M.Y.: Symbolic {LTLf}
  synthesis. In: IJCAI. pp. 1362--1369. AAAI Press (2017)

\end{thebibliography}
%
\clearpage 
\section*{Appendix}

\subsection{Automata over Finite and Infinite words}

A (nondeterministic) \emph{automaton} is a tuple $\A = (\alphabet, S , \init, \delta, F )$ 
where $\alphabet$ is a finite set of symbols (called an alphabet),
$S$ is a finite set of states,
$\init \in \states$ is the initial state,
$F \subseteq S$ is the set of accepting states, and
$\delta \subseteq S \times \Sigma \times S $ is the transition relation. 
An automaton on \emph{finite} words is called a \emph{nondeterministic finite-state automaton} (NFA), while an automaton over \emph{infinite} words is called a \emph{nondeterministic B\"uchi automaton} (NBA). 
An NFA is said to be \emph{deterministic} (DFA) if for each state $s$ and letter $a$, $ |\{s'|(s, a, s') \in \delta \textrm{ for some $s'$} \}|
\leq 1 $.
Deterministic B\"uchi automata (DBAs) are defined analogously.

Let $\A$ be an NFA.
For a finite word $w = w_0\cdots w_n \in \alphabet^*$, a {\em run} of $\A$ over $w$ is a finite state sequence $\rho = s_0\dots s_{n+1} \in S^{+}$ such that $s_0 = \init$ and for all $i\in \{0,\dots n\}$, $(s_i, w_i, s_{i+1}) \in \delta$ holds. A run $\rho = s_0\dots s_{n+1}$ is an {\em accepting run} if $s_{n+1} \in F$.
A word $w$ is \emph{accepted} by $\A$ if $\A$ has an accepting run over $w$.

Let $\B$ be an NBA.
Similarly, a run of $\B$ over an infinite word $w = w_0 w_1 \cdots \in \alphabet^{\omega}$ is an infinite sequence $\rho = s_0 s_1 \cdots \in S^{\omega}$ such that $s_0 = \init$ and for all $i \in \N$, $(s_i, w_i, s_{i+1}) \in \delta$.
Let $ \mathit{inf}(\rho) $ denote the set of states that occur infinitely often in run ${\rho}$. 
A run $\rho$ is an {\em accepting run} in $\B$ if $ \mathit{inf}(\rho)\cap F \neq \emptyset $.
An infinite word $w$ is accepted by $\B$ if $\B$ has an accepting run over $w$.

We denote by $\lang{\B}$ (resp. $\lang{\A}$) the set of all words accepted by $\B$ (resp. $\A$).
It is known that NFAs/DFAs recognize exactly \emph{regular} languages while NBAs accept exactly \emph{$\omega$-regular} languages. 
In the remainder of the paper, we denote by $w_i, i \geq 0$ the $i$-th element in the sequence $w$.

\subsection{Semantics of $\ltlf$ and $\ltl$}
\label{app:semantics}

We first give the semantics of $\ltlf$ formulas.
A finite sequence $\rho$ over $2^{\ap}$ is said to satisfy an $\ltlf$ formula $\phi$ over $\ap$, denoted by $\rho\models \phi$, if $\rho, 0 \models \phi$ where for all positions $0 \leq i < \size{\rho}$, $\rho, i\models \phi$ is defined inductively  on $\phi$ as follows:

\begin{itemize}
    \item $\rho, i \models \ltlftrue$,
    \item $\rho, i \not\models \ltlffalse$,
    \item $\rho, i \models a$ iff $a \in \rho_i$ where $\rho_i$ is the $i$-th element of $\rho$ for all $0\leq i<\size{\rho}$,
    \item $\rho, i \models \neg \phi$ iff $\rho, i \not\models \phi$,
    \item $\rho, i \models \phi_{1} \land \phi_{2}$ iff $\rho, i \models \phi_{1}$ and $\rho, i \models \phi_{2}$,
    \item $\rho, i \models \phi_{1} \lor \phi_{2}$ iff $\rho, i \models \phi_{1}$ or $\rho, i \models \phi_{2}$,
    \item $\rho, i \models \ltlfX \phi$ iff $i + 1 < \size{\rho}$ and $\rho, i + 1 \models \phi$,
    \item $\rho, i \models \phi_{1} \ltlfU \phi_{2}$ iff there exists $j$ s.t. $i \leq j < \size{\rho}$ and $\rho, j \models \phi_{2}$, and for all $k$, $i \leq k < j$, we have $\rho, k \models \phi_{1}$,
    \item $\rho, i \models \ltlfF \phi$ iff there exists $j$ s.t. $i \leq j < \size{\rho}$ and $\rho, j \models \phi$,
    \item $\rho, i \models \ltlfG \phi$ iff for all $j$ s.t. $i \leq j < \size{\rho}$, $\rho, j \models \phi$.
\end{itemize}

To obtain the semantics of $\ltl$ formulas, $\rho$ must be an infinite sequence.
Thus, the length of $\rho$, denoted as $\size{\rho}$, is $\infty$.
It actually means that we can just drop all restrictions that the integers need to be less than $\size{\rho}$ meant for $\ltlf$ semantics
We use a subscript $\infty$ for all $\ltl$ modalities to distinguish with their $\ltlf$ counterparts. 
For all positions $ i \geq 0$, $\rho, i\models \phi$ is defined inductively on $\phi$ as follows:
\begin{itemize}
    \item $\rho, i \models \ltlftrue$,
    \item $\rho, i \not\models \ltlffalse$,
    \item $\rho, i \models a$ iff $a \in \rho_i$ where $\rho_i$ is the $i$-th element of $\rho$ for all $i \geq 0$,
    \item $\rho, i \models \neg \phi$ iff $\rho, i \not\models \phi$,
    \item $\rho, i \models \phi_{1} \land \phi_{2}$ iff $\rho, i \models \phi_{1}$ and $\rho, i \models \phi_{2}$,
    \item $\rho, i \models \phi_{1} \lor \phi_{2}$ iff $\rho, i \models \phi_{1}$ or $\rho, i \models \phi_{2}$,
    \item $\rho, i \models \ltlfX_{\infty} \phi$ iff $\rho, i + 1 \models \phi$,
    \item $\rho, i \models \phi_{1} \ltlU \phi_{2}$ iff there exists $j$ s.t. $j \geq i$ and $\rho, j \models \phi_{2}$, and for all $k$, $i \leq k < j$, we have $\rho, k \models \phi_{1}$,
    \item $\rho, i \models \ltlF \phi$ iff there exists $j$ s.t. $j \geq i$ and $\rho, j \models \phi$,
    \item $\rho, i \models \ltlG \phi$ iff for all $j\geq i$ s.t. $j \geq i$, $\rho, j \models \phi$.
\end{itemize}

\subsection{Proof of Theorem~\ref{thrm:PropertiesPref}}

\noindent
\textbf{Theorem~\ref{thrm:PropertiesPref}.}
For $\ltlf$ formula $\phi$, let $\pref(\phi)$ be as defined above. Then,
\begin{enumerate}
    \item{\label{thrm:propertiesPref:safety}} $\pref(\phi)$ is a safety language.
    \item {\label{thrm:propertiesPref:omega}} $\pref(\phi)$ is $\omega$-regular. NBA representing $\pref(\phi)$ consists of $2^{2^{\mathcal{O}(|\phi|)}}$ states.
\end{enumerate}

\paragraph{Proof of ~\autoref{thrm:PropertiesPref}-~\ref{thrm:propertiesPref:safety}.}  

A language $ L \subseteq \Sigma^\omega$ is a {\em safety language} if for every word $w \notin L$ there exists a finite-prefix $u$ of $w$ such for all $y \in \Sigma^\omega$ the word $u\cdot y \notin L$. Such prefixes are referred to as {\em bad prefix}.  


Consider $w \in \Sigma^\omega$ such that $w \notin \pref(\phi)$. 
By definition of $\pref(\phi)$, there exists an $n>0$ s.t. the finite-prefix $w[0,n]\models \neg \phi$. Clearly,  every  infinite extensions of $w[0,n]$ will also not be contained in $\pref(\phi)$, i.e. for all $y \in \Sigma^\omega$, 
 $w[0,n]\cdot y \notin \pref(\phi)$. Hence, $\pref(\phi)$  is a safety language. 
 
\paragraph{Proof of ~\autoref{thrm:PropertiesPref}-~\ref{thrm:propertiesPref:omega}.}  
Given $\ltlf$ formula $\phi$, the NBA for $\pref(\phi)$ can be constructed as follows:
\begin{enumerate}
    \item Construct a DFA $D = (\alphabet, Q, \init, \delta, F)$ for $\neg\phi$, i.e.,  $\mathcal{L}(D) = \mathcal{L}(\neg\phi)$.
    
    We require $D$ to be \emph{complete} in the sense that for every state $s$ and every alphabet $a \in \alphabet$, there exists a successor $t = \delta(s, a)$.
    
    \item Obtain a DBA $C = (\alphabet, Q, \init, \delta', F)$ by converting all accepting states $F$ of $D$ to accepting sink states in $C$. For this, replace all outgoing transitions from all accepting states in $D$ with self loops on all letters.
    
    Formally, replace every $\delta(f,a) = t$ in DFA $D$ with $f = \delta'(f, a)$ in DBA $C$, for all $f\in F$ and $a \in \alphabet$. For all other states, let $\delta'$ behaves identically to $\delta$. 
    
    \item Obtain the desired NBA $B = (\alphabet, Q, \init, \delta',  \Final = Q\setminus F)$ by swapping accepting and non-accepting states of $C$.
\end{enumerate}
Since $C$ is a DBA with all accepting states as sink states, swapping accepting and non-accepting states results in its complementation. Hence, it is sufficient to show that $\mathcal{L}(C)$ accepts the complement of $\pref(\phi)$. In other words, $C$ accepts $w \in \Sigma^\omega$ iff there exists a finite-prefix of $w$ that satisfies $\neg\phi$. 
Clearly,  $w \in \lang{C}$ then $w$ must have a finite-prefix satisfying $\neg\phi$ since the accepting states of $C$ and $D$ are identical and all but outgoing transitions from accepting states are retained.
Conversely, let $w\in\Sigma^\omega$ such that it contains a finite prefix that satisfies $\neg\phi$. 
Despite $\delta$ and $\delta'$ being different, we need to show that $w$ is accepted. Let $v$ be the shortest prefix of $w$ satisfying $\neg\phi$. Since $D$ is a DFA, $v$ has a unique run in $D$. This run also appears in $C$ because all transitions appearing in this run in $D$ are retained in $C$ as none of them are outgoing transitions from accepting states (if it weren't so, then $v$ would not have been the shortest prefix of $w$ that satisfies $\neg\phi$).
Further, since accepting states in $C$ are sink states, $w \in \lang{C}$.
Finally, the number of states of $C$ are bounded by those of $D$ which is doubly exponential in $|\phi|$~\cite{de2013linear}. 
\qed

\subsection{Proof of Lemma~\ref{lemma:prefixlangauge}}

{\textbf Lemma~\ref{lemma:prefixlangauge}
Let $L_n$ and $F_n$ be as defined above. Then $$L_n \uplus \{0,1,\#\}^\omega = \pref(F_n \uplus \{0,1,\#\}^*).$$

\begin{proof}
    First, we show that $L_n \uplus \{0,1,\#\}^\omega \subseteq \pref(F_n \uplus \setnocond{0,1,\#}^{*})$. Trivially, all prefixes of words in $\{0,1,\#\}^\omega$  are contained in $\setnocond{0,1,\#}^*$ since ``$\&$" does not appear in any of them. It remains to show that $L_n \subseteq \pref(F_n \uplus \setnocond{0,1,\#}^{*})$. Let $u\cdot \& \cdot v \in L_n$. We establish that all prefixes of 
    $u\cdot \& \cdot v$ are contained in $F_n \uplus \setnocond{0,1,\#}$. We perform case analysis of prefixes:
    \begin{enumerate}
        \item When the prefix is a prefix of $u$. These prefixes are contained in $\setnocond{0,1,\#}^*$ since ``$\&$" does not appear in the prefix. 
        \item When prefix of is of the form $u\cdot\&$. Now, $u\cdot\& \in F_n$ since it contains exactly one ``$\&$" and the end of $u\cdot\&$ is not in the form $\#w\#$ for $w \in \{0,1\}^n$.
        \item When prefix is of the form $u\cdot \& \cdot y$ but $y$ does not end in $\# w \#$ for $w \in \{0,1\}^n$. For the same reason as above, $u\cdot \& \cdot y \in F_n$.
        \item When prefix is of the form $u\cdot \& \cdot y$ and $y$ ends in $\# w \#$ for $w \in \{0,1\}^n$. Since $u\cdot \& \cdot v \in L_n$, we know that every $\#w\#$ appearing in $v$ ``$\&$" must have appeared in $u$, for $w \in \{0,1\}^n$. Since $y$ is a prefix of $v$ and  $\#w\#$ is at the end of $y$, we get that  $\#w\#$ must have also appeared in $u$. Hence,   $u\cdot \& \cdot y \in F_n$.
    \end{enumerate}
    Hence, $L_n \uplus \{0,1,\#\}^\omega \subseteq \pref(F_n \uplus \setnocond{0,1,\#}^*)$.

Next, we prove  $ \pref(F_n \uplus \setnocond{0,1,\#}^*) \subseteq L_n \uplus \{0,1,\#\}^\omega$. First, observe that for $x \in \pref(F_n \uplus \setnocond{0,1,\#}^*)$, $x$ can contain at most one occurrence of ``$\&$".  By case analysis:
\begin{enumerate}
    \item If $x$ does not contain ``$\&$", then clearly, $x \in \{0,1,\#\}^\omega$.

    \item Otherwise, the word is of the form $u\cdot \& \cdot v$ where $u \in \{0,1,\#\}^*$ and $v \in \{0,1,\#\}^\omega$. Either there are no occurrences of $\#w\#$ in $v$, for $w \in \{0,1\}^n$. In this case, $u\cdot \& \cdot v \in L_n$ vacuously. 
    
    Otherwise, there are occurrences of $\#w\#$ in $v$. Let $u \cdot \& \cdot y$ be an arbitrary prefix of  $u\cdot \& \cdot v$ that ends in $\#w\#$. Since  $u \cdot \& \cdot y \in F_n \uplus \setnocond{0,1,\#}^*$, $u \cdot \& \cdot y \in F_n$. Thus, $\#w\#$ must have appeared in $u$ as well. Finally, since there are only finitely many possibilities of words of the form $\# w \#$, we conclude that every occurrence of $\# w \#$ in $v$ must have also appeared in $u$. Hence, $u\cdot \& \cdot v  \in L_n$.
\end{enumerate}
Hence,  $ \pref(F_n \uplus \setnocond{0,1,\#}^*) \subseteq L_n \uplus \{0,1,\#\}^\omega$.

Therefore $\pref(L_{\psi_n}) = L_n \uplus \setnocond{0,1,\#}^{\omega}$.
\qed
\end{proof}

\subsection{Encoding of (R1) from Theorem~\ref{thrm:prefixdouble}}

\begin{align*}
\mathsf{OnlyOneProp}&:= \ltlfG(0 \rightarrow \neg 1 \land \neg \& \land \neg \#) \land \ltlfG(1 \rightarrow \neg 0 \land \neg \& \land \neg \#) \\ 
& \land \ltlfG(\& \rightarrow \neg 0 \land \neg 1 \land \neg \#)  \land \ltlfG(\# \rightarrow \neg 0 \land \neg 1 \land \neg \&) \land \ltlfG(0 \lor 1\lor \& \lor \#).
\end{align*}

\subsection{Proof of ~\autoref{lem:ltlf-f-g}}

\textbf{\autoref{lem:ltlf-f-g}}
    Let $\phi \in \ltlf_{\setminus \setnocond{\ltlfR, \lor}}$ and let $\ltl$ $t(\phi)$ be  as defined above.
    Then, $\lang{t(\phi)} = \pref(\phi)$ and $\mathcal{O}(|\phi|) = \mathcal{O}(|t(\phi)|)$.

\begin{proof}
    Trivially, $\mathcal{O}(|\phi|) = \mathcal{O}(|t(\phi|)$ holds. 
    We prove that $\lang{t(\phi)} = \pref(\phi)$ by structural induction on $\phi$.
    Let $w = w_0 w_1 \cdots\in \alphabet^{\omega}$ where $w_i$ is the $i$-th letter in $w$.
    Recall, $w[0,n]$ denotes the subsequence $w_0 \cdots w_{n-1}$ of $w$ for $n > 0$. Then
    \begin{itemize}
        \item $\phi = \ell$ (resp. $\phi = \neg \ell$).
        By definition, $t(\phi) = \ell$. It is trivial that $w \in \pref(\phi) = \pref(\ell)$ iff $w \in \lang{\ell}$ since either $w_0 \models \ell$ or $w_0 \not\models \ell$.
        
        \item $\phi = \psi_1 \land \psi_2$.
        Then $t(\phi) = t(\psi_1) \land t(\psi_2)$.
        Assume that $w \in \lang{t(\psi_1) \land t(\psi_2)}$.
        By LTL semantics, $w \in \lang{t(\psi_1)}$ and $w \in \lang{t(\psi_2)}$.
        It follows that $w \in \pref(\psi_1)$ and $w \in \pref(\psi_2)$, based on induction assumption.
        It means that for all $ i> 0$, $w[0,i] \models \psi_1$ and $w[0,i] \models \psi_2$.
        Thus, $w[0,i] \models \psi_1 \land \psi_2$ for all $i > 0$.
        We then have that $w \in \pref(\phi)$.

        Assume that $w \in \pref(\phi) = \pref(\psi_1 \land \psi_2)$.
        It follows that for all $i > 0$, $w[0,i] \models \psi_1 \land \psi_2$, i.e., $w \in \pref(\psi_1)$ and $w\in \pref(\psi_2)$.
        By induction assumption, we have that $w \in \lang{t(\psi_1)}$ and $w \in \lang{t(\psi_2)}$.
        Consequently, $w \models t(\psi_1) \land t(\psi_2)$, i.e., $w \in \lang{t(\psi_1) \land t(\psi_2)}$.
        
        \item $\phi = \ltlfF \psi$. 
        Then $t(\phi) = t(\psi)$.
        By induction assumption, we have that $w \in \pref(\psi)$ iff $w \in \lang{t(\psi)}$.

        Assume that $w \in \lang{t(\phi)} = \lang{t(\psi)}$, i.e., $w \in \pref(\psi)$.
        It follows that for every $i > 0$, $w[0,i] \models \psi$.
        Obviously, for every $i > 0$, $w[0, i] \models \ltlfF \psi$.
        Consequently, $w \in \pref(\phi)$.
        
        Assume that $w \in \pref(\phi)$.
        Then for every $i> 0$, $w[0,i] \models \phi = \ltlfF \psi$.
        By semantics of $\ltlf$, $w[0,1] \models \psi$, i.e., $w_0 \models \psi$.
        It follows that for every $i > 0$, we also have that $w[0, i] \models  \psi$, indicating that $w \in \pref(\psi)$.
        By induction assumption, $w \in \lang{t(\psi)} = \lang{t(\phi)}$.
        So we are done for this case.

        \item $\phi = \psi_1 \ltlfU \psi_2$.
        Then $t(\phi) = t(\psi_2)$.
        The proof is quite similar to the one for $\ltlfF \psi$.
    By induction assumption, we have that $w \in \pref(\psi_2)$ iff $w \in \lang{t(\psi_2)}$.
        
        Assume that $w \in \lang{t(\phi)} = \lang{t(\psi_2)}$, i.e., $w \in \pref(\psi_2)$.
        It follows that for every $i > 0$, $w[0,i] \models \psi_2$.
        Obviously, for every $i > 0$, $w[0, i] \models \psi_1 \ltlfU \psi_2$ since $w[0,1] \models \psi_2$.
        Consequently, $w \in \pref(\phi)$.
        
        Assume that $w \in \pref(\phi)$.
        Then for every $i> 0$, $w[0,i] \models \phi = \psi_1 \ltlfU \psi_2$.
        By semantics of $\ltlf$, $w[0,1] \models \psi_2$, i.e., $w_0 \models \psi_2$.
        It follows that for every $i > 0$, we also have that $w[0, i] \models  \psi_2$, indicating that $w \in \pref(\psi_2)$.
        By induction assumption, $w \in \lang{t(\psi_2)} = \lang{t(\phi)}$.
        So we are done for this case.
        
        \item $\phi = \ltlfG \psi$.
        Then, $t(\phi) = \ltlG (t(\psi))$.
        By induction assumption, we have that we have that $w \in \pref(\psi)$ iff $w \in \lang{t(\psi)}$.

        Assume that $w \in \lang{t(\phi)} = \lang{\ltlG (t(\psi))}$.
        By semantics of $\ltl$, for every $i \geq 0, w[i, \infty] \in \lang{t(\psi)}$.
        In other words, we have that $w[i, \infty] \in \pref(\psi)$ for all $i\geq 0$.
        It follows that $w[i,i+1] \models \psi$ for all $i \geq 0$, according to definition of $\pref$ languages.
        Then we have that $w[0,i] \models \ltlfG \psi$ for all $i \geq 0$ in $\ltlf$ semantics.
        Obviously, $w \in \pref(\ltlfG \psi)$.

        Assume that $w \in \pref(\ltlfG \psi)$.
        By definition of $\pref$ languages, we have that $w[0, i] \models \ltlfG \psi$ for all $i > 0$.
        By semantics of $\ltlf$, we have $w_i \models \psi$ for all $i \geq 0$ (The last position of the word needs to satisfy $\psi$);
        Also, $w[i, j] \models \psi$ for all $j > i$.
        By definition of $\pref$ languages, we have that $w[i, \infty] \in \pref(\psi)$ for all $i \geq 0$.
        Based on induction assumption, we have $w[i, \infty] \in \lang{t(\psi)}$ for all $i \geq 0$.
        It follows that $w \models \ltlG {t(\psi)}$.
        Thus, we have done for this case. 
        
        \item $\phi = \ltlfN \psi$. Then $t(\phi) = \ltlfX_{\infty} t(\psi)$.
        By induction assumption, $w \in \pref(\psi)$ iff $w \in \lang{t(\psi)}$.
        
        Assume that $w \in \pref(\phi)$. Then $w[0, i] \models \ltlfN \psi$ for all $i > 0$, including $i = 2$. By $\ltlf$ semantics, it follows that $w[1,2] \models \psi$ since $w_0$ is not the last position when $i = 2$.
        It follows that we have that $w[1, i] \models \psi$ for all $i > 1$.
        So, $w[1, \infty] \in \pref(\psi)$, i.e., $w[1,\infty] \in \lang{t(\psi)}$ based on induction assumption.
        Then we have $w \models \ltlfX_{\infty} t(\psi)$, i.e., $w \in \lang{\ltlfX_{\infty} t(\psi)}$.

        Assume that $w \in \lang{\ltlfX_{\infty} t(\psi)}$. Then $w[1, \infty] \models t(\psi)$.
        By induction assumption, we have $w[1, \infty] \in \pref(\psi) $.
        It follows that $w[1, 2] \models \psi$ by definition of $\pref$ languages.
        Then $w[0,2] \models \ltlfN \psi$.
        Clearly, we have $w[0, i] \models \ltlfN \psi$ for all $i > 0$, including when $i = 1$.
        Consequently, we have $w \in \pref(\phi).$

        \item $\phi = \ltlfX \psi$.
        Then $t(\phi) = \ltlffalse$.
        It is impossible for a word $w \in \pref(\ltlfX \psi)$ to hold since $w[0,1] \not\models \ltlfX \psi$ as there is no next position at position $0$.
        Therefore, $\lang{\pref(\phi)} = \lang{\ltlffalse} = \emptyset$ since there are no words satisfying $\ltlffalse$.
        \end{itemize}
    \qed
\end{proof}

\subsection{Missing details from Section~\ref{sec:lowerbound}}

\subsubsection{Consistency conditions (C1) and (C2)}
We present the encoding of the first two consistency conditions (C1) and (C2). Recall, we require the following two:
\begin{enumerate}
    \item [(C1).] Consistency within a configuration (the binary encoding of each cell's position is correct)
    \item [(C2).] Consistency with the input word (the first configuration is correct)
\end{enumerate}

Condition (C1) only needs to reason about adjacent cells in the same configuration. If $(b_1, \ldots, b_{cn})$ and $(b'_1, \ldots, b'_{cn})$ are the binary encodings of the positions of two adjacent cells, and $\mathit{Succ}(b_1, \ldots, b_{cn}, b'_1, \ldots, b'_{cn})$ is a propositional formula capturing that $(b'_1, \ldots, b'_{cn})$ encodes the successor (mod $2^{cn}$) of $(b_1, \ldots, b_{cn})$ (see below for details), then the formula
\begin{equation*}
    \ltlfG\big((part_0 \land \ltlfX^{2cn+1}\,true) \rightarrow \mathit{Succ}(\ltlfX^1\,bit, \ldots, \ltlfX^{cn}\,bit, \ltlfX^{cn + 2}\,bit, \ldots, \ltlfX^{2cn + 1}\,bit)\big)
\end{equation*}
expresses that if we start at the beginning of the encoding of a cell ($part_0$) and the prefix is long enough to include the entirety of the successor cell ($\ltlfX^{2cn + 1}\,\ltlftrue$), then $\mathit{Succ}$ holds between the encodings of the two cells (note that $b_i$ is given by $\ltlfX^{i}\,bit$ and $b'_i$ is given by $\ltlfX^{cn + 1 + i}\,bit$). Similarly, the formula $\ltlfX^{cn}\,\ltlftrue \rightarrow \bigwedge^{cn}_{i=1} \ltlfX^i \neg bit$ expresses that the encoding of the first cell's position is $0$.

Condition (C2) only requires looking at the $n$ cell contents that should contain the input word in the first configuration, plus ensuring that all other cells on the tape are blank. Checking the cells that should contain the input word can be expressed by the formula
$$ \ltlfX^{(cn + 1)n}\,\ltlftrue \rightarrow \Big(\bigwedge^n_{i=1} \ltlfX^{(cn + 1)i}\,cell_{x_i}\Big) $$
meaning that if the prefix is long enough to cover all $n$ cells ($\ltlfX^{(cn + 1)n}\,\ltlftrue$), then the content of the $i$-th cell is $x_i$ ($\ltlfX^{(cn + 1)i}\,cell_{x_i}$), for all $i$ from $1$ to $n$. Ensuring that all other cells are blank can likewise be expressed by a formula of polynomial size (see below for details).

\paragraph{Consistency within a configuration.}

As explained above, the first consistency condition can be represented by a conjunction of the formula $\ltlfX^{cn}\,\ltlftrue \rightarrow \bigwedge^{cn}_{i=1} \ltlfX^i \neg bit$, which expresses that the encoding of the first cell's position is $0$, and the formula $\ltlfG((part_0 \land \ltlfX^{2cn+1}\,\ltlftrue) \rightarrow \mathit{Succ}(\ltlfX^1\,bit, \ldots, \ltlfX^{cn}\,bit, \ltlfX^{cn + 2}\,bit, \ldots, \ltlfX^{2cn + 1}\,bit))$, which expresses that the encoded position of each successive cell is the successor of the previous one. The propositional formula $\mathit{Succ}$ can be defined as
$$\mathit{Succ}(b_1, \ldots, b_{cn}, b'_1, \ldots, b'_{cn}) = (b'_1 \leftrightarrow \neg b_1) \land \bigwedge^{cn}_{i=2} (b'_i \leftrightarrow (b_i \oplus (b_{i-1} \land \neg b'_{i-1})))$$
which expresses the successor relation between two binary numbers $b_{cn}\ldots b_1$ and $b'_{cn}\ldots b'_1$ (note that we consider $b_1$ the least significant digit). The subformula $(b'_1 \leftrightarrow \neg b_1)$ expresses that the least significant digit is flipped, while $(b'_i \leftrightarrow (b_i \oplus (b_{i-1} \land \neg b'_{i-1})))$ (where $\oplus$ is the exclusive-or operator) expresses that the $i$-th digit is flipped if there is a carry (which only happens when the $(i-1)$-th digit has flipped from $1$ to $0$).

\paragraph{Consistency with the input word.}

The second consistency condition is composed of two formulas. As explained above, the formula $\ltlfX^{(cn + 1)n}\,\ltlftrue \rightarrow \left(\bigwedge^n_{i=1} \ltlfX^{(cn + 1)i}\,cell_{x_i}\right)$ expresses that the first $n$ cells of the first configuration contain the input word $x = x_1 \ldots x_n$. The second formula ensures that all other cells are blank, and can be expressed by
$$\ltlfX^{(cn + 1)(n + 1)}\,\ltlftrue \rightarrow \ltlfX^{(cn + 1)(n + 1)}\Big((part_0 \rightarrow cell_{\emptyset})\,\text{W}\,\big(part_0 \land \bigwedge^{cn}_{i = 1} \ltlfX^i\neg bit \big)\Big)$$
meaning that if the prefix is long enough to reach the $(n+1)$-th cell ($\ltlfX^{(cn + 1)(n + 1)}\,\ltlftrue$), the contents of every cell from this point on must be blank ($part_0 \rightarrow cell_{\emptyset}$) until we reach a new configuration, indicated by the encoding of the cell position resetting back to $0$ ($part_0 \land \bigwedge^{cn}_{i = 1} \ltlfX^i\neg bit$). Note that the ``zeroth'' cell (the cell where the head starts, immediately before the input word) is also blank, but this is enforced by the transition relation of $M$.

\subsubsection{Missing formulas for (C3) }

\begin{itemize}
    \item $\varphi_\delta$ expresses that if the head is on $c$ ($cell_{(q, \gamma)}$), then in $c_{Last}$ the head must have moved to a different cell and written the appropriate symbol $\gamma'$ given by the transition relation ($L^{-cn}\,cell_{\gamma'}$)
    \[
        \varphi_\delta \equiv \bigwedge_{q \in Q} \bigwedge_{\gamma \in \Gamma} \big(cell_{(q, \gamma)} \rightarrow L^{-cn}\,cell_{\gamma'}\big)
        \tag*{where $\delta(q, \gamma) = (q', \gamma', d)$}
    \]
    
    \item $\varphi^{\leftarrow}_\delta$ expresses that if the head is on the cell to the \emph{right} of $c$ ($\ltlfX^{cn + 1}\,cell_{(q, \gamma_2)}$), and the transition relation requires it to move left, then in the next configuration the head must have moved to $c_{Last}$ ($L^{-cn}\,cell_{(q', \gamma_1)})$)

    \[
    \varphi^{\leftarrow}_\delta \equiv \bigwedge_{q \in Q} \bigwedge_{\gamma_1 \in \Gamma} \bigwedge_{\gamma_2 \in \Gamma}\Big(\big(cell_{\gamma_1} \land \ltlfX^{cn + 1}\,cell_{(q, \gamma_2)}\big) \rightarrow L^{-cn}\,cell_{(q', \gamma_1)}\Big) \tag*{where $\delta(q, \gamma_2) = (q', \gamma', \leftarrow)$} 
    \]
    
    \item $\varphi^{\rightarrow}_\delta$ expresses that if the head is on the cell to the \emph{left} of $c$ ($cell_{(q, \gamma_1)}$), and the transition relation requires it to move right, then in the next configuration the head must have moved to $c_{Last}$ ($L^{-cn}\,cell_{(q', \gamma_2)})$)

    \[\varphi^{\rightarrow}_\delta \equiv \bigwedge_{q \in Q} \bigwedge_{\gamma_1 \in \Gamma} \bigwedge_{\gamma_2 \in \Gamma}\Big(\big(cell_{(q, \gamma_1)} \land \ltlfX^{cn + 1}\,cell_{\gamma_2}\big) \rightarrow L^{-cn}\,cell_{(q', \gamma_2)}\Big) \tag*{where $\delta(q, \gamma_1) = (q', \gamma', \rightarrow)$} 
    \]
    \item Finally, $\varphi^0_\delta$ expresses that if the head is neither on $c$ nor on the cells adjacent to it ($cell_{\gamma_1} \land \ltlfX^{cn + 1}\,cell_{\gamma_2} \land \ltlfX^{2(cn + 1)}\,cell_{\gamma_3}$), then the contents of the cell don't change ($L^{-cn}\,cell_{\gamma_2}$)
    \[
    \varphi^0_\delta \equiv \bigwedge_{\gamma_1 \in \Gamma} \bigwedge_{\gamma_2 \in \Gamma} \bigwedge_{\gamma_3 \in \Gamma} \Big(\big(cell_{\gamma_1} \land \ltlfX^{cn + 1}\,cell_{\gamma_2} \land \ltlfX^{2(cn + 1)}\,cell_{\gamma_3}\big) \rightarrow L^{-cn}\,cell_{\gamma_2}\Big)
    \]
\end{itemize}

Note that in the latter two formulas $c$ is the cell to the right of the current cell ($\ltlfX^{cn+1}\,\mathsf{MatchLastCell}$) this is necessary so that $\varphi^{\rightarrow}_\delta$ and $\varphi^0_\delta$ can refer to the cell to the left of $c$.


\subsection{NBA with at least $2^{2^{n}}$ states}

Let $n \in \N$ and $\Sigma = \{0,1,\#, \&\}$. 
Consider the language $L_n \subseteq \Sigma^ \omega$ where
\[u\cdot \& \cdot v \in L_n \text{ s.t. if }   \#w\#  \text{ appears in } u \text{ then } \#w\# \text{ also appears in } v,\]
where $w\in\{0,1\}^n$, $u \in \{0,1,\#\}^*$ and $v \in \{0,1,\#\}^\omega$.
Essentially, $L_n$ is  a bit-level adaption of the language $K_d$ where $x \cdot \& \cdot y \in K_D$ if digits appearing in $x$ are a subset of digits appearing in $y$, where $x \in D^*$ and  $y \in D^\omega$ for $D = \{0,1,\cdots, d-1\}$. We show that all NBA of $K_d$  consists of at least $2^D$ states. This proof can easily be adapted to show that all NBA of $L_n$ consists of $2^{2^{\Omega(n)}}$ states.

First, note that $K_d$ is a safety $\omega$-regular language. Let $C_d$ be an (non-deterministic B\"uchi) automaton representing $K_d$. Then, $C_d$ can be trimmed by removing all states that are unreachable from the initial state and at least one accepting state. Next, all states of the trimmed automaton can be converted to accepting states. Let us denote this automaton by $A_d$. Clearly, $L(A_d) = L(C_d)$ and $A_d$ has fewer states than $C_d$.

We claim that $C_d$ must have at least $2^d$ states. Suppose there are fewer than $2^d$ states.
We will use the notation $x_S$ and $y_T$ to denote  finite and infinite words over the digits $D$ s.t. $S$ and $T$ denote the set of digits appearing in $ x_S$ and $y_T$ respectively. For a state $Q$ in $A_d$ with outgoing transitions on $\&$, let $\& \cdot y_{T_1}, \dots \& \cdot y_{T_p}$ be all the infinite words with paths starting in $Q$. Since all paths are accepting ($A_d$ is a safety automaton), all finite words to $Q$ must be of the form $x_S$ where $S  \subseteq T$ and $T = \bigcap_{i=1}^p T_i$.  Now, consider a word  $x_T\& y_T$. We claim that all its accepting paths must pass through states of the form $Q$. Suppose $x_T$ has a path to a state $Q'$ with an outgoing transition on $\&$. Similar to $T$ for $Q$, let $T'$ be defined for $Q'$. We assume $T' \neq T$. Clearly, $T \subseteq T'$, since otherwise it would accept a word $x_T\& y_S$ where $T \nsubseteq S$. Furthermore, $T' \subseteq T$ since otherwise $\& \cdot y_T$ will not have a path from $Q'$. Hence, $T = T'$. Hence, for every $S\subseteq D$, $A_d$ must have at least one unique state to accept words of the form $x_S\& y_S$. Thus, $A_d$ must have at least $2^D$ states. Subsequently, all automata $C_d$ of the language must contain at least $2^D$ states.

\end{document}